\pdfoutput=1
\documentclass[10pt]{article}
\usepackage{xcolor}
\usepackage{graphicx}
\usepackage{amsmath,amsfonts,amssymb,graphics,amsthm}
\usepackage{hyperref}
\usepackage{comment}
\usepackage{tabularx}
\usepackage[protrusion=true,expansion=true]{microtype}
\usepackage{enumerate}
\usepackage{bbm}
\usepackage{mathrsfs}
\usepackage[margin=1in]{geometry}
\usepackage[shortlabels]{enumitem}
\usepackage{cleveref}

\hypersetup{
    colorlinks=false,
    linktocpage,
    }

\numberwithin{equation}{section}

\newtheorem{theorem}{Theorem}[section]
\newtheorem{corollary}[theorem]{Corollary}
\newtheorem{lemma}[theorem]{Lemma}
\newtheorem{proposition}[theorem]{Proposition}

\newtheorem{remark}[theorem]{Remark}
\newtheorem{definition}[theorem]{Definition}

\theoremstyle{remark}

\def\@rst #1 #2other{#1}
\newcommand\MR[1]{\relax\ifhmode\unskip\spacefactor3000 \space\fi
  \MRhref{\expandafter\@rst #1 other}{#1}}
\newcommand{\MRhref}[2]{\href{http://www.ams.org/mathscinet-getitem?mr=#1}{MR#2}}

\def\MR#1{\href{http://www.ams.org/mathscinet-getitem?mr=#1}{MR#1}}

\newcommand{\C}{\mathbbm{C}}
\newcommand{\D}{\mathbbm{D}}

\newcommand{\T}{\mathbbm{T}}

\newcommand{\bs}{\boldsymbol}

\newcommand{\dd}{\text{\rm d}} 

\let\Re\undefined
\DeclareMathOperator{\Re}{Re}

\def\det{\mathrm{det}}

\def\alb#1\ale{\begin{align*}#1\end{align*}}
\def\allb#1\alle{\begin{align}#1\end{align}}

\newcommand{\aryb}{\begin{eqnarray*}}
\newcommand{\arye}{\end{eqnarray*}}
\def\alb#1\ale{\begin{align*}#1\end{align*}}
\newcommand{\eqb}{\begin{equation}}
\newcommand{\eqe}{\end{equation}}
\newcommand{\eqbn}{\begin{equation*}}
\newcommand{\eqen}{\end{equation*}}

\let\Re\undefined

\def\bL{\mathbf{L}}

\def\alb#1\ale{\begin{align*}#1\end{align*}}
\def\allb#1\alle{\begin{align}#1\end{align}}

\def\alb#1\ale{\begin{align*}#1\end{align*}}

\newcommand{\Re}{\mathrm{Re}}

\let\originalleft\left
\let\originalright\right
\renewcommand{\left}{\mathopen{}\mathclose\bgroup\originalleft}
\renewcommand{\right}{\aftergroup\egroup\originalright}

\DeclareMathAlphabet{\mathpzc}{OT1}{pzc}{m}{it}

\title{Ward identities: a geometric point of view and applications}
\author{Baojun Wu  \quad\quad\quad        Shengjing Xu}
\date{May 2024}

\begin{document}

\maketitle

\begin{abstract}
The work of Baverez--Guillarmou--Kupiainen--Rhodes~\cite{semigroup} is the
starting point of this paper. It shows that analytic changes of boundary
parametrizations act differentiably on Liouville amplitudes, with derivative
given by Virasoro operators and a scalar anomaly term. We compute this scalar
term. Its holomorphic part is a Schwarzian boundary integral, which gives a
geometric explanation of the Virasoro central term. We then derive local Ward
identities on disks, annuli, and pairs of pants. They give finite recursions
for descendant matrix coefficients. From these recursions we recover the
polynomial factorization of normalized pair-of-pants coefficients. In the
annular zero-weight limit, we recover the Shapovalov form. For a pair of pants
with two incoming boundaries, we recover the formal chiral vertex-operator
coefficients. We also give a geometric proof of smoothness in the bulk
insertion points and derive the genus-zero arbitrary level BPZ equations for degenerate bulk
insertions.
\end{abstract}
\tableofcontents

\section{Introduction}

Liouville conformal field theory (LCFT), introduced by Polyakov in the study
of two-dimensional quantum gravity~\cite{polyakov-qg1}, is formally described
by a path integral over fields on a Riemannian surface. Probabilistic Liouville conformal field theory has been constructed on the
sphere, on compact Riemann surfaces, and in several boundary settings; see
\cite{dkrv-lqg-sphere,grv-higher-genus,hrv-disk,remy-annulus,
wu2023liouvilleconformalfieldtheory}. 
The basic observables of LCFT on a closed surface are the correlation functions
\[
 \left\langle\prod_{j=1}^nV_{\alpha_j}(x_j)
 \right\rangle_{\Sigma,g}.
\]
For a surface with analytically parametrized boundary, the natural object is
a Liouville amplitude. In Segal's formalism~\cite{Segal1988TheDO}, each
boundary component carries a copy of the state space, an amplitude is a
multilinear functional or an operator between tensor products of state
spaces, and sewing of surfaces corresponds to composition of amplitudes. The
probabilistic Liouville amplitudes and their sewing properties were
constructed in~\cite{segalaxiom}. Their analytic boundary deformations were
studied in~\cite{Baverez_2024,semigroup}, where differentiation of an
amplitude produces the Virasoro action together with a scalar contribution
coming from the Liouville anomaly functional.

\textbf{1. Schwarzian variation of the anomaly.}
The first purpose of this paper is to identify this scalar term
geometrically. Let $(\Sigma,g,\boldsymbol\zeta)$ be an admissible surface with
incoming boundary components and filling $(\widehat\Sigma,\widehat g)$. For a
vector field $v$ holomorphic near the boundary, write
\[
 \zeta_k^*v
 =-\sum_{n\in\mathbb Z}[v_k]_nw^{n+1}\partial_w.
\]
In genus zero, the holomorphic variation of the anomaly is expressed entirely
through the Schwarzian derivatives of the inverse boundary coordinates.

\begin{theorem}[Genus-zero Schwarzian variation]\label{thm:metric}
Let $(\Sigma,g,\boldsymbol\zeta)$ be an admissible genus-zero surface with
$b\geq1$ incoming boundary components, and let
$(\widehat\Sigma,\widehat g)\simeq(\widehat{\mathbb C},\widehat g)$ be its
filling. Let $v=v(z)\partial_z$ be meromorphic on $\widehat\Sigma$ and
holomorphic on a neighborhood of $\overline\Sigma$. If $f_t^v$ and
$f_t^{\mathrm iv}$ are the corresponding local conformal flows, set
$\Sigma_t^v:=f_t^v(\Sigma)$, $g_t^v:=(f_t^v)_*g$, and
\[
 S_{\rm L}^0(\Sigma_t^v)
 :=S_{\rm L}^0(\Sigma_t^v,g_t^v,\widehat g),
\]
with the analogous notation for $\mathrm iv$. Then
\begin{align}
 \frac12\left.
 \left(
 \partial_tS_{\rm L}^0(\Sigma_t^v)
 -\mathrm i\partial_tS_{\rm L}^0(\Sigma_t^{\mathrm iv})
 \right)\right|_{t=0}
 -\frac1{24}\sum_{k=1}^b[v_k]_0
 =
 \frac{\mathrm i}{24\pi}
 \sum_{k=1}^b
 \oint_{\partial_k\Sigma}^{\rm in}
 v(z)S_{\zeta_k^{-1}}(z)\,dz.
 \label{eq:introduction-Schwarzian-formula}
\end{align}
Here all contours are oriented by the incoming parametrizations.
\end{theorem}

The theorem is a special case of a boundary-local version, Proposition~\ref{prop:metric-extension}. If $v$ is given
only on boundary collars and $V$ is a smooth extension to the filling, then
one obtains the same Schwarzian contour term together with an interior
correction involving $\bar\partial V$. 


\textbf{2. Ward identities for matrix coefficients of amplitudes.}
Combining the above Schwarzian formula with the differentiability theorem for
Liouville amplitudes~\cite{semigroup}, Weyl covariance, and 
local Ward identity~\cite{reviewsegal}. The Virasoro modes on the boundary are related to a
covariant differential operator in the bulk marked points, while the central
term is precisely the Schwarzian contour integral above. We apply this
identity to the three genus-zero building blocks used in the sewing
decomposition: a disk with two interior marked points, an annulus with one
interior marked point, and a pair of pants.

Let $
  \psi_j:U_j\longrightarrow \widehat{\mathbb C},
  \qquad j=1,2,3$
be univalent holomorphic maps on neighborhoods $U_j$ of
$\overline{\mathbb D}$. Assume that the compact
sets $\psi_j(\overline{\mathbb D})$ are pairwise disjoint, and set
\[
  \mathcal P_{\boldsymbol\psi}
  :=\widehat{\mathbb C}\setminus
  \bigcup_{j=1}^3\psi_j(\mathbb D),
  \qquad
  \boldsymbol\psi:=(\psi_1,\psi_2,\psi_3),
\] These
recursions recover the polynomial factorization of normalized descendant
coefficients. More precisely, on the physical spectrum
$\alpha_j=Q+\mathrm iP_j$, set
\[
 \boldsymbol\Delta_{\boldsymbol\alpha}
 :=(\Delta_{\alpha_1},\Delta_{\alpha_2},\Delta_{\alpha_3}),
\]
and let
$\boldsymbol\nu=(\nu_1,\nu_2,\nu_3)$ and
$\widetilde{\boldsymbol\nu}
=(\widetilde\nu_1,\widetilde\nu_2,\widetilde\nu_3)$
be triples of partitions. Whenever the primary coefficient is non-zero,
there is a polynomial
$\omega_{\mathcal P_{\boldsymbol\psi},\boldsymbol\nu}$ in the three
conformal weights, normalized by
$\omega_{\mathcal P_{\boldsymbol\psi},\boldsymbol\emptyset}=1$, such that
\begin{align}
 &\frac{
 \mathcal A_{\mathcal P_{\boldsymbol\psi},g_{\rm ad},
 \boldsymbol\psi}
 \left(
 \Psi_{\alpha_1,\nu_1,\widetilde\nu_1},
 \Psi_{\alpha_2,\nu_2,\widetilde\nu_2},
 \Psi_{\alpha_3,\nu_3,\widetilde\nu_3}
 \right)}{
 \mathcal A_{\mathcal P_{\boldsymbol\psi},g_{\rm ad},
 \boldsymbol\psi}
 \left(
 \Psi_{\alpha_1},\Psi_{\alpha_2},\Psi_{\alpha_3}
 \right)}
=
 \omega_{\mathcal P_{\boldsymbol\psi},\boldsymbol\nu}
 (\boldsymbol\Delta_{\boldsymbol\alpha})
 \overline{
 \omega_{\mathcal P_{\boldsymbol\psi},
 \widetilde{\boldsymbol\nu}}
 (\boldsymbol\Delta_{\boldsymbol\alpha})}.
 \label{eq:introduction-descendant-factorization}
\end{align}
The statement extends to the analytic continuation domain of the generalized
amplitudes. The factorization itself was proved in~\cite{segalaxiom}. Our new
point is a direct geometric derivation from the local Ward identity. The
recursion removes the negative modes one at a time. Polynomiality and the
separation of the holomorphic and anti-holomorphic sectors are then explicit. For analytic parametrization, the matrix coefficients also contain the Schwarzian derivative of $\boldsymbol{\psi}$, which is a consequence of Theorem~\ref{thm:metric}.

The annulus recursion has a further consequence. When the interior weight
tends to zero, it identifies the normalized annulus coefficient with the
Shapovalov form. We also study a pair of pants with two incoming and one
outgoing boundary. After normalization, its matrix coefficients satisfy the
Virasoro intertwining relation and agree with the formal chiral
vertex-operator coefficients of~\cite[Section~17.1]{Teschner_2001}.

\textbf{3. Applications.} First, we give a geometric proof of smoothness of Liouville correlation
functions in the bulk marked points. This was previously proved on the sphere
in~\cite{Oikarinen_2019} and on compact Riemann surfaces in
\cite[Proposition~5.2]{oikarinen2022stressenergyliouvilleconformalfield}. We
insert a zero-weight field and remove a disk around it. The correlation
function is then an amplitude evaluated on the generalized state $\Psi_0$. A
weighted Ward identity expresses each marked-point derivative as a finite sum
of descendant amplitudes and explicit anomaly terms. Iteration gives
smoothness on the configuration space.

Second, the round-disk Ward identity turns the arbitrary-level null-vector relation, which was proved in~\cite{irreducible_rep} (to make the introduction simpler, we only recall the level-two case)
\[
 (\alpha^2\mathbf L_{-2}+\mathbf L_{-1}^2)\Psi_\alpha=0
\]
into the BPZ differential equation. In a flat affine coordinate, it reads

\begin{equation}
 \left[
 \frac1{\alpha^2}\partial_x^2
 +\sum_{i=1}^N
 \left(
 \frac1{x-x_i}\partial_{x_i}
 +\frac{\Delta_{\beta_i}}{(x-x_i)^2}
 \right)
 \right]
 \left\langle
 V_\alpha(x)\prod_{i=1}^NV_{\beta_i}(x_i)
 \right\rangle^{\rm flat}_{\widehat{\mathbb C}}
 =0.
 \label{eq:introduction-level-two-BPZ}
\end{equation}
The above second-order BPZ equation was proved in~\cite{krv-local}.
For a general degenerate weight $\alpha_{r,s}$, we apply the same disk
recursion to each monomial in the singular vector. This gives a differential
operator of order at most $rs$, with principal part $\partial_x^{rs}$ in flat
coordinates. Our method also applies to the BPZ equation for LCFT on the upper half-plane with bulk degenerate insertions. We also remark that if one of the $r$ and $s$ equals $1$, the BPZ equation was studied in~\cite{zhu2020higherorderbpzequations} via the Coulomb gas integrals method.

\textbf{Organization of the paper.}
The paper is organized as follows. Section~2 recalls Liouville amplitudes, the
Liouville Hilbert space, descendant states, sewing, and the BGKR variational
formula. Section~3 proves the Schwarzian variation of the anomaly. Section~4
derives the local Ward identities and the results on descendant matrix
coefficients. Section~5 contains the smoothness and BPZ applications.

\subsubsection*{Acknowledgements:} The authors would like to thank Colin Guillarmou for suggesting the question concerning a geometric approach to Ward identities, as well as for his valuable discussions. B. Wu has been supported by the National Key R\&D Program of China (No.\ 2023YFA1010700).

\section{Preliminaries}\label{sec:preliminaries}

Throughout the paper, we fix $\gamma\in(0,2)$ and use the conventions
\begin{equation}\label{eq:liouville-parameters}
 Q:=\frac{\gamma}{2}+\frac{2}{\gamma},\qquad
 c_{\rm L}:=1+6Q^2,
 \qquad
 \Delta_\alpha:=\frac{\alpha}{2}\left(Q-\frac{\alpha}{2}\right).
\end{equation}
This section collects analytic, probabilistic, and representation-theoretic results related to Liouville CFT. The probabilistic construction of Liouville CFT was initiated on the sphere in~\cite{dkrv-lqg-sphere}, extended to higher genus in~\cite{grv-higher-genus}, and developed for surfaces with boundary in~\cite{hrv-disk,remy-annulus,wu2023liouvilleconformalfieldtheory}. We use the amplitude formalism and its sewing properties from~\cite{segalaxiom,guillarmou2024conformalbootstrapsurfacesboundary}, the Virasoro construction from~\cite{Baverez_2024,semigroup}, and the description of the degenerate Liouville modules from~\cite{irreducible_rep}.

In the following, we recall the probabilistic construction of the Liouville field, the vertex operators, and the amplitudes. Next, we describe the Liouville representation of the Virasoro algebra. Finally, we state the variational formula for the amplitudes with connect to the Virasoro algebra.

\subsection{Riemann surfaces with analytic boundary}

A \emph{Riemann surface with analytic boundary} is a compact oriented smooth surface $\Sigma$ whose boundary is a disjoint union of circles
\[
 \partial\Sigma=\bigsqcup_{j=1}^{b}\partial_j\Sigma,
\]
together with a complex structure on the interior which extends analytically across every boundary component. Equivalently, every point of $\partial\Sigma$ has a neighborhood which is mapped biholomorphically to a neighborhood of an interval in the closed upper half-plane, with the boundary mapped to the real axis. The complex structure induces the orientation of $\Sigma$ and the corresponding boundary orientation.

A smooth Riemannian metric $g$ on $\Sigma$ is \emph{compatible} with the complex structure if, in every holomorphic coordinate $z$, it has the form $g=e^{\sigma(z,\bar z)}|dz|^2$. We denote by $dv_g$ the volume measure, by $d\ell_g$ the boundary measure, by $K_g$ the scalar curvature, and by $k_g$ the geodesic curvature. 

Let $\mathbb T:=\{z\in\mathbb C:|z|=1\}$. An analytic boundary parametrization is a real-analytic diffeomorphism $\zeta_j:\mathbb T\to\partial_j\Sigma$ which extends biholomorphically to an annular neighborhood of $\mathbb T$.

\begin{definition}[Incoming and outgoing boundary parametrizations]\label{def:orientation}
Let $\nu$ be the inward unit normal along $\partial_j\Sigma$ and set $\dot\zeta_j(\theta):=\partial_\theta\zeta_j(e^{\mathrm i\theta})$. The parametrized boundary $\partial_j\Sigma$ is called \emph{outgoing} if the ordered pair $(\dot\zeta_j,\nu)$ is positively oriented in $\Sigma$, and \emph{incoming} otherwise.
\end{definition}

We denote
\(
 \boldsymbol\zeta=(\zeta_1,\ldots,\zeta_{b}),
\)
for the boundary parametrization. 
\begin{definition}[Admissible metric and admissible surface]\label{def:admissible surface}
A compatible metric $g$ is \emph{admissible} for $(\Sigma,\boldsymbol\zeta)$ if for every Dirichlet boundary component, there exists $\delta<1$ such that
\begin{equation}\label{eq:admissible-collar}
 g=(\zeta_j^{-1})^*\frac{|dz|^2}{|z|^2}
\end{equation}
on the collar parametrized by $\zeta_j(\mathbb A_{\delta,\delta^{-1}})$, where
\[
 \mathbb A_{r,R}:=\{z\in\mathbb C:r<|z|<R\}.
\]
The triple $(\Sigma,g,\boldsymbol\zeta)$ is then called an \emph{admissible surface}. If interior and boundary marked points with weights are also specified, we speak of an admissible surface with insertions.
\end{definition}

The cylindrical condition in~\eqref{eq:admissible-collar} is precisely the condition which permits smooth sewing of two metrics along oppositely oriented parametrized boundary circles. These conventions agree with those used in~\cite{segalaxiom,semigroup,guillarmou2024conformalbootstrapsurfacesboundary}.

\subsection{Probabilistic Liouville theory}

\subsubsection{Gaussian free field and Gaussian multiplicative chaos}

We use the non-negative Laplacian $\Delta_g=d^*d$. If $\partial\Sigma\neq\emptyset$, let $\Delta_{g,D}$ be its Dirichlet realization, with domain $H^2(\Sigma)\cap H_0^1(\Sigma)$. This operator has compact resolvent and strictly positive discrete spectrum. We denote its zeta-regularized determinant by $\det(\Delta_{g,D})$.

The Dirichlet Green operator $R_{g,D}:=2\pi\Delta_{g,D}^{-1}$ is characterized by
\begin{equation}\label{eq:dirichlet-resolvent}
 \Delta_{g,D}R_{g,D}=2\pi\mathrm{Id}
 \quad\text{on }L^2(\Sigma,dv_g),
\end{equation}
and its integral kernel is denoted by $G_{g,D}$. Since the range of $R_{g,D}$ is contained in the domain of $\Delta_{g,D}$, the Dirichlet boundary condition is built into this definition. If $(\lambda_{D,j},e_{D,j})_{j\geq1}$ is an orthonormal basis of real Dirichlet eigenfunctions, the Dirichlet Gaussian free field is
\begin{equation}\label{eq:dirichlet-gff}
 X_{g,D}:=\sqrt{2\pi}\sum_{j\geq1}
 \frac{a_j}{\sqrt{\lambda_{D,j}}}e_{D,j},
\end{equation}
where $(a_j)_{j\geq1}$ are independent standard real Gaussian variables. The series converges almost surely in $H^s(\Sigma)$ for every $s<0$, and its covariance is $\mathbb E[X_{g,D}(x)X_{g,D}(y)]=G_{g,D}(x,y)$ in the sense of distributions. In the boundary-field framework used below, we fix $s\in(-\frac12,0)$; see~\cite{semigroup,segalaxiom}.

For a boundary field $\widetilde{\boldsymbol\varphi}=(\widetilde\varphi_1,\ldots,\widetilde\varphi_b)$, let $P_{\Sigma,\boldsymbol\zeta}\widetilde{\boldsymbol\varphi}$ denote its harmonic extension, characterized for smooth boundary data by
\begin{equation}\label{Pdefin}
 \Delta_gP_{\Sigma,\boldsymbol\zeta}
 \widetilde{\boldsymbol\varphi}=0,
 \qquad
 P_{\Sigma,\boldsymbol\zeta}
 \widetilde{\boldsymbol\varphi}|_{\partial_j\Sigma}
 =
 \widetilde\varphi_j\circ\zeta_j^{-1},
 \qquad j=1,\ldots,b.
\end{equation}
For the distributional boundary fields used in the amplitude construction, this extension is understood by continuity from smooth data. The Liouville field on a surface with boundary is then
\begin{equation}\label{eq:liouville-field-boundary}
 \phi_g
 :=
 X_{g,D}
 +
 P_{\Sigma,\boldsymbol\zeta}
 \widetilde{\boldsymbol\varphi}.
\end{equation}

When $\partial\Sigma=\emptyset$, let $\Pi_0$ be the orthogonal projection onto the constants and let $R_g$ be the mean-zero Green operator determined by $\Delta_gR_g=2\pi(\mathrm{Id}-\Pi_0)$ and $R_g1=0$. We denote its kernel by $G_g$, write $X_g$ for the centered Gaussian free field with covariance $G_g$, and use $\det{}'(\Delta_g)$ for the zeta-regularized determinant with the zero eigenspace omitted. In this case we set
\begin{equation}\label{eq:liouville-field-zero-mode}
 \phi_g:=c+X_g,
 \qquad c\in\mathbb R.
\end{equation}

Let $X$ denote $X_{g,D}$ when $\partial\Sigma\neq\emptyset$ and $X_g$ when $\Sigma$ is closed. We write $X_\varepsilon$ for its circle-average regularization in the metric $g$. Uniformly for $x$ in compact subsets of $\Sigma^\circ$,
\begin{equation}\label{eq:gff-variance-asymptotic}
 \mathbb E[X_\varepsilon(x)^2]
 =
 -\log\varepsilon+W_g(x)+o(1),
\end{equation}
where $W_g$ denotes the corresponding Robin function. Equivalent smooth local regularizations give the same limiting chaos measure.

For $\gamma\in(0,2)$, the regularized bulk Gaussian multiplicative chaos measure is
\begin{equation}\label{eq:bulk-gmc}
 M_{\gamma,g,\varepsilon}(X,dx)
 :=
 \varepsilon^{\gamma^2/2}
 e^{\gamma X_\varepsilon(x)}\,dv_g(x).
\end{equation}
It converges in probability, for the topology of weak convergence of Radon measures, to a non-trivial random measure denoted by $M_{\gamma,g}(X,dx)$. For a deterministic shift $h$, we use the convention
\[
 M_{\gamma,g}(X+h,dx)
 :=
 e^{\gamma h(x)}M_{\gamma,g}(X,dx);
\]
in particular, when $\partial\Sigma\neq\emptyset$, this defines $M_{\gamma,g}(\phi_g,dx)$ from~\eqref{eq:liouville-field-boundary}. We refer to~\cite{rhodes-vargas-review} for Gaussian multiplicative chaos and to~\cite{dkrv-lqg-sphere,segalaxiom} for its use in Liouville theory.

For an interior point $x\in\Sigma^\circ$ and a weight $\alpha\in\mathbb R$, the regularized bulk vertex operator is
\begin{equation}\label{eq:bulk-vertex}
 V_{\alpha,g,\varepsilon}(x)
 :=
 \varepsilon^{\alpha^2/2}
 e^{\alpha\phi_{g,\varepsilon}(x)}.
\end{equation}
The pair $(x,\alpha)$ is called a \emph{bulk insertion}. These regularized insertions are used inside the limiting amplitudes and correlation functions defined below. Their conformal weight is $\Delta_\alpha$ as fixed in~\eqref{eq:liouville-parameters}.

\subsubsection{The Liouville Hilbert space and amplitudes}
\label{sub:hilbert}
We recall the Hilbert space for the Liouville CFT and Dirichlet-boundary amplitude formalism which is discussed in \cite{segalaxiom,semigroup, reviewsegal}.

\paragraph{The circle field and the Hilbert space.}
The centered Gaussian free field on the unit circle is the random
Fourier series
\begin{equation}\label{eq:circle-gff}
 \varphi(\theta)
 =
 \sum_{n\neq0}\varphi_ne^{\mathrm i n\theta},
 \qquad
 \varphi_n
 =
 \frac{x_n+\mathrm i y_n}{2\sqrt n}
 \quad(n>0),
 \qquad
 \varphi_{-n}
 =
 \overline{\varphi_n},
\end{equation}
where $(x_n,y_n)_{n\geq1}$ are independent standard Gaussian pairs.
It has covariance
$\mathbb E[\varphi(\theta)\varphi(\theta')]
=-\log|e^{\mathrm i\theta}-e^{\mathrm i\theta'}|$
and belongs almost surely to $H^s(\mathbb T)$ for every $s<0$.
We write $\widetilde\varphi=c+\varphi$, set
$\Omega_{\mathbb T}:=(\mathbb R^2)^{\mathbb N^*}$, and denote by
\begin{equation}\label{eq:circle-gaussian-measure}
 \mathbb P_{\mathbb T}
 :=
 \bigotimes_{n\geq1}
 \frac1{2\pi}
 e^{-\frac12(x_n^2+y_n^2)}
 \,dx_n\,dy_n
\end{equation}
the law of the non-constant modes. On the circle, we use the
normalized Hermitian product
\begin{equation}\label{eq:circle-inner-product}
 \langle f,h\rangle_2
 :=
 \frac1{2\pi}
 \int_0^{2\pi}
 f(e^{\mathrm i\theta})
 \overline{h(e^{\mathrm i\theta})}
 \,d\theta,
\end{equation}
and on a disjoint union of circles we take the direct sum of these
products. The reference measure and the Liouville Hilbert space are
\begin{equation}\label{eq:liouville-hilbert}
 d\mu_0(\widetilde\varphi)
 :=
 dc\otimes d\mathbb P_{\mathbb T}(\varphi),
 \qquad
 \mathcal H
 :=
 L^2(\mathbb R\times\Omega_{\mathbb T},\mu_0).
\end{equation}
We denote the Hermitian product by
$\langle\cdot,\cdot\rangle_{\mathcal H}$ and, for later use, write
$c_-:=\min(c,0)$ and $c_+:=\max(c,0)$. Weighted spaces in the zero
mode are always understood with respect to
$dc\otimes\mathbb P_{\mathbb T}$.

\paragraph{Liouville amplitudes.}
Let
$(\Sigma,g,\mathbf x,\boldsymbol\alpha,\boldsymbol\zeta)$
be a compact Riemann surface with analytic parametrized boundary,
where
$\mathbf x=(x_1,\ldots,x_m)\subset\Sigma^\circ$
are pairwise distinct and
$\boldsymbol\alpha=(\alpha_1,\ldots,\alpha_m)\in\mathbb R^m$.
We assume $\mu>0$ and $\alpha_j<Q$ for every $j$. If
$\partial\Sigma\neq\emptyset$, all boundary components are Dirichlet
and
$\boldsymbol\zeta=(\zeta_1,\ldots,\zeta_b)$
parametrizes them.

For boundary data
$\widetilde{\boldsymbol\varphi}$, let
$P_{\Sigma,\boldsymbol\zeta}\widetilde{\boldsymbol\varphi}$
and $\phi_g$ be as in
\eqref{Pdefin} and
\eqref{eq:liouville-field-boundary}.
The Liouville weight is
\begin{equation}\label{eq:liouville-weight}
 \mathcal W_{\Sigma,g}(\phi)
 :=
 \exp\left(
 -\frac{Q}{4\pi}
 \int_\Sigma K_g\phi\,dv_g
 -\frac{Q}{2\pi}
 \int_{\partial\Sigma}k_g\phi\,d\ell_g
 -\mu M_{\gamma,g}(\phi,\Sigma)
 \right).
\end{equation}

Let $\mathbf D_\Sigma$ be the Dirichlet-to-Neumann operator of
$\Sigma$, transported to $\mathbb T^b$ by
$\boldsymbol\zeta$, and let $\mathbf D$ be the direct sum of the
Dirichlet-to-Neumann operators of the unit disk. The free-field
amplitude is
\begin{equation}\label{eq:free-field-amplitude}
 \mathcal A^0_{\Sigma,g}
 (\widetilde{\boldsymbol\varphi})
 :=
 \exp\left(
 -\frac12
 \left\langle
 \widetilde{\boldsymbol\varphi},
 (\mathbf D_\Sigma-\mathbf D)
 \widetilde{\boldsymbol\varphi}
 \right\rangle_2
 \right).
\end{equation}
For an admissible metric,
$\mathbf D_\Sigma-\mathbf D$ is smoothing in the boundary
coordinates. Hence~\eqref{eq:free-field-amplitude}, initially
defined for smooth boundary data, extends continuously to the
distributional boundary fields sampled by $\mu_0^{\otimes b}$.
We use the determinant normalization
\[
 Z_{\Sigma,g,D}
 :=
 \det(\Delta_{g,D})^{-1/2}
 \exp\left(
 \frac1{8\pi}
 \int_{\partial\Sigma}k_g\,d\ell_g
 \right).
\]

For a continuous nonnegative functional $F$ of the field, the
Liouville amplitude is
\begin{align}
 &\mathcal A_{\Sigma,g,\mathbf x,\boldsymbol\alpha,
 \boldsymbol\zeta}
 (F,\widetilde{\boldsymbol\varphi})
 :=
 Z_{\Sigma,g,D}\,
 \mathcal A^0_{\Sigma,g}
 (\widetilde{\boldsymbol\varphi})
 \lim_{\varepsilon\to0}
 \mathbb E\left[
 F(\phi_g)
 \prod_{j=1}^{m}
 V_{\alpha_j,g,\varepsilon}(x_j)
 \mathcal W_{\Sigma,g}(\phi_g)
 \right],
 \label{eq:liouville-amplitude-prob}
\end{align}
where the expectation is over the Dirichlet Gaussian free field
$X_{g,D}$. If $F$ is bounded and the global
Seiberg inequality in~\eqref{def: seiberg bound} holds, then the
amplitude belongs to $L^2(\mu_0^{\otimes b})$ and therefore
determines the Hilbert-space tensor used under the Segal axiom \cite{segalaxiom}.

For a closed surface, with $\phi_g$ as in
\eqref{eq:liouville-field-zero-mode}, define $\phi_g:=c+X_g$, and use
\[
 Z_{\Sigma,g}
 :=
 \left(
 \frac{\det{}'(\Delta_g)}
 {\operatorname{vol}_g(\Sigma)}
 \right)^{-1/2}.
\]
In this case,
\begin{align}
 \mathcal A_{\Sigma,g,\mathbf x,\boldsymbol\alpha}(F)
 :=Z_{\Sigma,g}
 \lim_{\varepsilon\to0}
 \int_{\mathbb R}
 \mathbb E\Bigg[
 &F(c+X_g)
 \prod_{j=1}^{m}
 V_{\alpha_j,g,\varepsilon}(x_j)
 \nonumber\\
 &\times
 \exp\left(
 -\frac{Q}{4\pi}
 \int_\Sigma K_g(c+X_g)\,dv_g
 -\mu M_{\gamma,g}(c+X_g,\Sigma)
 \right)
 \Bigg]dc.
 \label{eq:closed-liouville-amplitude}
\end{align}

\begin{definition}[Liouville correlation functions]
\label{def:correlation}
For a closed surface, we have
$$
\left\langle
 \prod_{j=1}^{m}V_{\alpha_j,g}(x_j)
 \right\rangle_{\Sigma,g}
=
\mathcal A_{\Sigma,g,\mathbf x,\boldsymbol\alpha}(1).
$$
\end{definition}
The correlation function is finite and non-zero precisely under the
Seiberg bounds
\begin{equation}\label{def: seiberg bound}
 \alpha_j<Q
 \quad(1\leq j\leq m),
 \qquad
 \sum_{j=1}^{m}\alpha_j
 >
 Q\chi(\Sigma).
\end{equation}
For a surface with parametrized boundary, the local conditions
$\alpha_j<Q$ suffice for the almost-everywhere definition of the
kernel, while the second inequality
in~\eqref{def: seiberg bound} gives its $L^2$ integrability.

\paragraph{Weyl covariance and diffeomorphism invariance.}
We use the Liouville anomaly functional
\begin{equation}\label{def:liouville action}
 S_{\rm L}^0(\Sigma,g_0,e^\omega g_0)
 :=
 \frac1{96\pi}
 \int_\Sigma
 \left(
 |d\omega|_{g_0}^2
 +2K_{g_0}\omega
 \right)dv_{g_0}.
\end{equation}
The following form of the covariance statement is the one used
below.

\begin{proposition}[Weyl covariance and diffeomorphism invariance]
\label{prop:weyl+diff}
Let $g'=e^\omega g$. If $\partial\Sigma\neq\emptyset$, assume that
$g$ and $g'$ are admissible for the same boundary parametrizations
and that $\omega$ vanishes on a neighborhood of
$\partial\Sigma$. Then
\begin{align}
 &\mathcal A_{\Sigma,g',\mathbf x,\boldsymbol\alpha,
 \boldsymbol\zeta}
 (F,\widetilde{\boldsymbol\varphi})
 =
 \exp\left(
 c_{\rm L}S_{\rm L}^0(\Sigma,g,g')
 -\sum_{j=1}^{m}
 \Delta_{\alpha_j}\omega(x_j)
 \right)
 \mathcal A_{\Sigma,g,\mathbf x,\boldsymbol\alpha,
 \boldsymbol\zeta}
 \left(
 F\left(\,\cdot-\frac Q2\omega\right),
 \widetilde{\boldsymbol\varphi}
 \right).
 \label{eq:weylcoa}
\end{align}
For a closed surface, the same identity holds without boundary data.
Moreover, if
$\psi:\Sigma\to\Sigma'$
is an orientation-preserving diffeomorphism, then
\begin{equation}\label{eq:diffeomorphism-covariance}
 \mathcal A_{\Sigma,\psi^*g,\mathbf x,\boldsymbol\alpha,
 \boldsymbol\zeta}
 (F,\widetilde{\boldsymbol\varphi})
 =
 \mathcal A_{\Sigma',g,\psi(\mathbf x),\boldsymbol\alpha,
 \psi\circ\boldsymbol\zeta}
 (F_\psi,\widetilde{\boldsymbol\varphi}),
\end{equation}
where
$F_\psi(\phi):=F(\phi\circ\psi)$.
\end{proposition}

\begin{proof}
The general Dirichlet-boundary Weyl formula is in
\cite[Proposition~2.5]{semigroup}. Since $\omega$ vanishes near the
boundary, its boundary trace and all boundary correction terms
vanish, and the formula reduces to~\eqref{eq:weylcoa}.
Diffeomorphism invariance is the second assertion of the same
proposition.
\end{proof}

\subsection{Liouville representation of Virasoro algebra}
We next recall how the Virasoro algebra acts on the Liouville Hilbert space. We first introduce the free-field representation, then the interacting operators, and finally the Liouville highest-weight modules. This order makes explicit the passage from the Gaussian boundary field to the representation-theoretic objects used in the Ward identities.

\subsubsection{The Liouville Hamiltonian and Virasoro operators}
Let $\mathcal S$ be the space of smooth cylinder functions of the Gaussian variables $(x_n,y_n)_{n\geq1}$, with polynomial growth together with all derivatives, and set
\(
 \mathcal C_\infty
 :=\operatorname{Span}\{\psi(c)F:\psi\in C_c^\infty(\mathbb R),\ F\in\mathcal S\}.
\)
This is a dense subspace in $\mathcal H$. For $n>0$, define
\[
 \partial_n:=\sqrt n(\partial_{x_n}-\mathrm i\partial_{y_n}),
 \qquad
 \partial_{-n}:=\sqrt n(\partial_{x_n}+\mathrm i\partial_{y_n}).
\]
On $\mathcal C_\infty$, introduce the Heisenberg operators
\begin{align*}
 \mathbf A_n&:=\frac{\mathrm i}{2}\partial_n,
 &\mathbf A_{-n}&:=\frac{\mathrm i}{2}(\partial_{-n}-2n\varphi_n),\\
 \widetilde{\mathbf A}_n&:=\frac{\mathrm i}{2}\partial_{-n},
 &\widetilde{\mathbf A}_{-n}&:=\frac{\mathrm i}{2}(\partial_n-2n\varphi_{-n}),
\end{align*}
and $\mathbf A_0=\widetilde{\mathbf A}_0:=\frac{\mathrm i}{2}(\partial_c+Q)$. Normal ordering places the annihilation operators to the right. The free-field Virasoro operators are
\begin{align}
 \mathbf L_n^0
 &:=-\mathrm i(n+1)Q\mathbf A_n
 +\sum_{m\in\mathbb Z}:\mathbf A_{n-m}\mathbf A_m:,
 \label{virassorotilde}\\
 \widetilde{\mathbf L}_n^0
 &:=-\mathrm i(n+1)Q\widetilde{\mathbf A}_n
 +\sum_{m\in\mathbb Z}:\widetilde{\mathbf A}_{n-m}\widetilde{\mathbf A}_m:.
 \nonumber
\end{align}
On $\mathcal C_\infty$, they satisfy two commuting Virasoro relations.

Let
\(
 V(\varphi):=\int_{0}^{2\pi}e^{\gamma\varphi(\theta)}\,d\theta
\)
be the circle GMC potential. Define the closed non-negative quadratic form
\begin{align}
 \mathcal Q(F)
 :=\int_{\mathbb R}\Bigg(&\frac12\|\partial_cF\|_{L^2(\Omega_{\mathbb T})}^2
 +\frac{Q^2}{2}\|F\|_{L^2(\Omega_{\mathbb T})}^2
 +\sum_{n\geq1}\|\mathbf A_nF\|_{L^2(\Omega_{\mathbb T})}^2+\sum_{n\geq1}\|\widetilde{\mathbf A}_nF\|_{L^2(\Omega_{\mathbb T})}^2
 +\mu e^{\gamma c}\langle VF,F\rangle_{L^2(\Omega_{\mathbb T})}\Bigg)dc.
 \label{eq:liouville-quadratic-form}
\end{align}
For $\gamma<\sqrt2$, $V$ is a multiplication operator; for $\gamma\geq\sqrt2$, the last term is defined by the corresponding closed quadratic-form limit. We denote the domain for quadratic form by $\mathcal D(\mathcal Q)$ and its anti-dual by $\mathcal D'(\mathcal Q)$.

For $n\geq0$, the interacting Virasoro operators are the operators in $\mathcal L\bigl(\mathcal D(\mathcal Q),\mathcal D'(\mathcal Q)\bigr)$ given by
\begin{align}
 \mathbf L_n
 &=\mathbf L_n^0+\frac\mu2e^{\gamma c}
 \int_0^{2\pi}e^{\mathrm i n\theta}e^{\gamma\varphi(\theta)}\,d\theta,\label{eq:interacting-virasoro}\\
 \widetilde{\mathbf L}_n
 &=\widetilde{\mathbf L}_n^0+\frac\mu2e^{\gamma c}
 \int_0^{2\pi}e^{-\mathrm i n\theta}e^{\gamma\varphi(\theta)}\,d\theta.
 \nonumber
\end{align}
For $n>0$, the negative modes are defined by duality,
\[
 \langle\mathbf L_{-n}F,G\rangle_{\mathcal H}
 :=\langle F,\mathbf L_nG\rangle_{\mathcal H},
 \qquad
 \langle\widetilde{\mathbf L}_{-n}F,G\rangle_{\mathcal H}
 :=\langle F,\widetilde{\mathbf L}_nG\rangle_{\mathcal H}.
\]
They satisfy the commutation relations
\begin{align}
 [\mathbf L_n,\mathbf L_m]
 &=(n-m)\mathbf L_{n+m}
 +\frac{c_{\rm L}}{12}(n^3-n)\delta_{n,-m}\,\mathrm{Id},
 \label{eq:virasoro-commutator}\\
 [\widetilde{\mathbf L}_n,\widetilde{\mathbf L}_m]
 &=(n-m)\widetilde{\mathbf L}_{n+m}
 +\frac{c_{\rm L}}{12}(n^3-n)\delta_{n,-m}\,\mathrm{Id},\nonumber\\
 [\mathbf L_n,\widetilde{\mathbf L}_m]&=0.\nonumber
\end{align} 
Moreover, the Liouville Hamiltonian is given by
\(
 \mathbf H=\mathbf L_0+\widetilde{\mathbf L}_0.
\)

\subsubsection{Generalized eigenstates and Liouville modules}
Set $c_-:=\min(c,0)$. For $\beta>0$, we use the weighted
form space
\begin{equation}\label{eq:weighted-form-domain}
 e^{-\beta c_-}\mathcal D(\mathcal Q)
 :=
 \left\{
 F:e^{\beta c_-}F\in\mathcal D(\mathcal Q)
 \right\},
 \qquad
 \|F\|_{e^{-\beta c_-}\mathcal D(\mathcal Q)}
 :=
 \|e^{\beta c_-}F\|_{\mathcal D(\mathcal Q)}.
\end{equation}
The generalized primary state $\Psi_\alpha$ is not, in general, an
element of $\mathcal H$. It forms an analytic family
\(
 \alpha \to \Psi_\alpha
\)
which is holomorphic on
$\{|\Re(\alpha)-Q|<\beta\}$ with values in
$e^{-\beta c_-}\mathcal D(\mathcal Q)$. The same statement holds for
all descendant families
$\alpha\mapsto\Psi_{\alpha,\nu,\widetilde\nu}$
\cite[Theorem~1.2]{Baverez_2024}. 
The initial
Poisson--scattering construction of this family is given in
\cite[Sections~6--7]{gkrv-bootstrap}, while the absence of poles of
the full Poisson operator is proved in
\cite[Theorem~1.1]{irreducible_rep}.

The primary state satisfies
\begin{equation}\label{eq:primary-highest-weight}
 \mathbf L_n\Psi_\alpha
 =
 \widetilde{\mathbf L}_n\Psi_\alpha=0
 \quad(n>0),
 \qquad
 \mathbf L_0\Psi_\alpha
 =
 \widetilde{\mathbf L}_0\Psi_\alpha
 =
 \Delta_\alpha\Psi_\alpha.
\end{equation}
In particular, for $\alpha=Q+\mathrm iP$ with $P\geq0$, one has
$\mathbf H\Psi_{Q+\mathrm iP}
=\frac12(Q^2+P^2)\Psi_{Q+\mathrm iP}$.

A partition is a finite non-increasing sequence
\[
 \nu=(\nu_1,\ldots,\nu_{\ell(\nu)}),
 \qquad
 \nu_1\geq\cdots\geq\nu_{\ell(\nu)}\geq1.
\]
Its size and length are
\(
 |\nu|:=\sum_{j=1}^{\ell(\nu)}\nu_j,
 \ell(\nu):=\#\{j:\nu_j\geq1\}.
\)
The empty partition is denoted by $\emptyset$ and has size and
length zero. For simplicity, we always set
\begin{equation}\label{eq:partition-operator}
 \mathbf L_{-\nu}
 :=
 \mathbf L_{-\nu_1}\cdots
 \mathbf L_{-\nu_{\ell(\nu)}},
 \qquad
 \widetilde{\mathbf L}_{-\widetilde\nu}
 :=
 \widetilde{\mathbf L}_{-\widetilde\nu_1}\cdots
 \widetilde{\mathbf L}_{-\widetilde\nu_{\ell(\widetilde\nu)}},
\end{equation}
with $\mathbf L_{-\emptyset}
=\widetilde{\mathbf L}_{-\emptyset}
=\mathrm{Id}$. Moreover, we denote
\[
 \mathcal T_n:=\{\nu\in\mathcal T:|\nu|=n\}, \qquad \mathcal T=\cup_n \mathcal T_n.
\]
Then we define
the two-sided descendants by
\begin{equation}\label{formula:eigenstate}
 \Psi_{\alpha,\nu,\widetilde\nu}
 :=
 \mathbf L_{-\nu}
 \widetilde{\mathbf L}_{-\widetilde\nu}
 \Psi_\alpha,
 \qquad
 \alpha\in\mathbb C,
 \quad
 \nu,\widetilde\nu\in\mathcal T.
\end{equation}
For every fixed pair
$(\nu,\widetilde\nu)$, the map
$\alpha\mapsto\Psi_{\alpha,\nu,\widetilde\nu}$ is analytic on
$\mathbb C$; see
\cite[Theorem~1.2]{Baverez_2024} and
\cite[Theorem~1.1]{irreducible_rep}. The descendants satisfy
\begin{align}
 \mathbf L_0\Psi_{\alpha,\nu,\widetilde\nu}
 &=
 (\Delta_\alpha+|\nu|)
 \Psi_{\alpha,\nu,\widetilde\nu},
 \nonumber\\
 \widetilde{\mathbf L}_0
 \Psi_{\alpha,\nu,\widetilde\nu}
 &=
 (\Delta_\alpha+|\widetilde\nu|)
 \Psi_{\alpha,\nu,\widetilde\nu},
 \label{eq:descendant-eigenvalues}\\
 \mathbf H\Psi_{\alpha,\nu,\widetilde\nu}
 &=
 (2\Delta_\alpha+|\nu|+|\widetilde\nu|)
 \Psi_{\alpha,\nu,\widetilde\nu}.
 \nonumber
\end{align}
For $\alpha\in\mathbb C$, define the Liouville
module
\begin{equation}\label{eq:liouville-module-definition}
 \mathcal V_\alpha
 :=
 \operatorname{Span}
 \left\{
 \Psi_{\alpha,\nu}
 :=
 \mathbf L_{-\nu}\Psi_\alpha:
 \nu\in\mathcal T
 \right\},
\end{equation}
which carries a Virasoro algebra action as stated above.

We next recall the contravariant form used below. Let
$M(c_{\rm L},\Delta_\alpha)$ be the Verma module with highest-weight
vector $v_\alpha$, and let $V(c_{\rm L},\Delta_\alpha)$ denote its
irreducible quotient. For $n\geq0$, set
$\mathcal T_n:=\{\nu\in\mathcal T:|\nu|=n\}$.

\begin{definition}[Shapovalov form]
\label{def:shapovalov-form}
The Shapovalov form
$\langle\cdot,\cdot\rangle_\alpha^{\rm Sh}$ is the unique symmetric
bilinear form on $M(c_{\rm L},\Delta_\alpha)$ such that
$\langle v_\alpha,v_\alpha\rangle_\alpha^{\rm Sh}=1$ and
$\langle\mathbf L_{-k}u,v\rangle_\alpha^{\rm Sh}
=\langle u,\mathbf L_kv\rangle_\alpha^{\rm Sh}$ for every
$k\geq1$. Its level-$n$ Gram matrix is
\begin{equation}\label{eq:shapovalov-matrix}
 F_\alpha^n(\nu_1,\nu_2)
 :=
 \left\langle
 \mathbf L_{-\nu_1}v_\alpha,
 \mathbf L_{-\nu_2}v_\alpha
 \right\rangle_\alpha^{\rm Sh},
 \qquad
 \nu_1,\nu_2\in\mathcal T_n.
\end{equation}
The coefficients $F_\alpha^n(\nu_1,\nu_2)$ are polynomials in
$c_{\rm L}$ and $\Delta_\alpha$.
\end{definition}

For $r,s\in\mathbb N^*$, set
\begin{equation}\label{eq:kac-weights}
 \alpha_{r,s}
 :=
 Q-r\frac\gamma2-s\frac2\gamma
 =
 (1-r)\frac\gamma2+(1-s)\frac2\gamma,
\end{equation}
and define
$\mathrm{Kac}^-:=\{\alpha_{r,s}:r,s\in\mathbb N^*\}$,
$\mathrm{Kac}^+
:=\{2Q-\alpha_{r,s}:r,s\in\mathbb N^*\}$, and
$\mathrm{Kac}:=\mathrm{Kac}^-\sqcup\mathrm{Kac}^+$.

\begin{proposition}[Liouville highest-weight modules {\cite[Theorem~1.2]{irreducible_rep}}]
\label{liouville module}
For every $\alpha\in\mathbb C$, the algebraic descendant space
$\mathcal V_\alpha$ is a quotient of
$M(c_{\rm L},\Delta_\alpha)$. More precisely:
\begin{enumerate}[(i)]
 \item if $\alpha\notin\mathrm{Kac}$, then
 $\mathcal V_\alpha$ is isomorphic to the Verma module
 $M(c_{\rm L},\Delta_\alpha)$ and is irreducible;

 \item if $\alpha\in\mathrm{Kac}^-$, then
 $\mathcal V_\alpha$ is isomorphic to the irreducible quotient
 $V(c_{\rm L},\Delta_\alpha)$;

 \item if $\alpha\in\mathrm{Kac}^+$, then
 $\mathcal V_\alpha$ is isomorphic to
 $V(c_{\rm L},\Delta_\alpha)$ unless
 \[
  \alpha\in
  \left(Q+\frac2\gamma\mathbb N^*\right)
  \cup
  \left(Q+\frac\gamma2\mathbb N^*\right),
 \]
 in which case $\mathcal V_\alpha=\{0\}$.
\end{enumerate}
\end{proposition}

For $\alpha=\alpha_{r,s}\in\mathrm{Kac}^-$, the Verma module
contains a singular vector at level $rs$, and its image in the
irreducible Liouville quotient vanishes. Consequently, there are
universal coefficients $\sigma_\nu^{r,s}$, depending only on
$c_{\rm L}$ and $\Delta_{\alpha_{r,s}}$, such that
\begin{equation}\label{eq:null-vector-relation}
 \mathbf S_{\alpha_{r,s}}\Psi_{\alpha_{r,s}}=0,
 \qquad
 \mathbf S_{\alpha_{r,s}}
 :=
 \sum_{|\nu|=rs}
 \sigma_\nu^{r,s}\mathbf L_{-\nu}.
\end{equation}
The normalization is fixed by
$\sigma_{(1,\ldots,1)}^{r,s}=1$. At level two,
\begin{equation}\label{eq:level-two-null-vector}
 \left(
 \alpha^2\mathbf L_{-2}
 +\mathbf L_{-1}^2
 \right)\Psi_\alpha=0,
 \qquad
 \alpha\in\{\alpha_{1,2},\alpha_{2,1}\}.
\end{equation}
This is the representation-theoretic input for the BPZ equations
used later.

For $\alpha=Q+\mathrm iP$ with $P>0$, one has
$\alpha\notin\mathrm{Kac}$, so
Proposition~\ref{liouville module} identifies $\mathcal V_\alpha$
with the Verma module and~\eqref{eq:shapovalov-matrix} with the
contravariant Gram matrix of the Liouville descendants. After
orthogonalizing each finite-dimensional level with these Gram
matrices, the two-sided descendants enter the spectral resolution of
$\mathbf H$; see~\cite{gkrv-bootstrap,Baverez_2024} and
\cite[Proposition~6.3]{segalaxiom}.

\subsection{Variation of amplitudes}

We now pass to the geometric interpretation of Liouville amplitudes with its variational formula. which gives the Virasoro generators. 

Let $(\Sigma_1,g_1,\mathbf x_1,\boldsymbol\alpha_1,\boldsymbol\zeta_1)$ and $(\Sigma_2,g_2,\mathbf x_2,\boldsymbol\alpha_2,\boldsymbol\zeta_2)$ be admissible surfaces. Assume that $k$ outgoing boundary components of $\Sigma_1$ are identified with $k$ incoming boundary components of $\Sigma_2$ through their parametrizations and that the cylindrical metrics glue smoothly. We denote the sewn surface by $(\Sigma,g,\mathbf x,\boldsymbol\alpha,\boldsymbol\zeta)=(\Sigma_1\#\Sigma_2,g_1\#g_2,\mathbf x_1\sqcup\mathbf x_2,\boldsymbol\alpha_1\sqcup\boldsymbol\alpha_2,\boldsymbol\zeta)$.

\begin{proposition}[Sewing of amplitudes {\cite[Propositions~5.1--5.2]{segalaxiom}}] \label{prop: gluing}
For nonnegative measurable functionals $F_j$ on the fields of $\Sigma_j$, the amplitude of the sewn surface is obtained by contracting the boundary variables on the glued components:
\begin{align}
 &\mathcal A_{\Sigma_1\#\Sigma_2,g_1\#g_2}
 (F_1\otimes F_2,
  \widetilde{\boldsymbol\varphi}_1,
  \widetilde{\boldsymbol\varphi}_2)
 =C_k\int
 \mathcal A_{\Sigma_1,g_1}
 (F_1,\widetilde{\boldsymbol\varphi}_1,
  \widetilde{\boldsymbol\varphi})
 \mathcal A_{\Sigma_2,g_2}
 (F_2,\widetilde{\boldsymbol\varphi},
  \widetilde{\boldsymbol\varphi}_2)
 \,d\mu_0^{\otimes k}(\widetilde{\boldsymbol\varphi}),
 \label{eq:amplitude-sewing}
\end{align}
where
\[
 C_k=
 \begin{cases}
  (2\pi)^{-k/2},&\partial\Sigma\neq\emptyset,\\[2mm]
  \sqrt2\,(2\pi)^{-(k-1)/2},&\partial\Sigma=\emptyset.
 \end{cases}
\]
If the $k$ outgoing factors of the first operator are precisely the $k$ incoming factors of the second, this is the operator identity $\mathcal A_{\Sigma}=C_k\,\mathcal A_{\Sigma_2}\circ\mathcal A_{\Sigma_1}$, after the natural permutation of tensor factors.
\end{proposition}

\subsubsection{Annular amplitudes and Virasoro generators}
\label{sec:Variation of annuli}

Fix $\delta<1$ and set $1+\varepsilon:=\delta^{-1}$. Let $\operatorname{Hol}_\delta(\mathbb A)$ be the Fr\'echet space of holomorphic functions on $\mathbb A_{\delta,\delta^{-1}}$, with topology defined by the seminorms
\begin{equation}\label{eq:hol-annulus-seminorm}
 \|F\|_{\varepsilon,k}^2
 :=\sum_{n\in\mathbb Z}
 (1+\varepsilon)^{2|n|+2}(1+|n|)^{2k}|F_n|^2,
 \qquad
 F(z)=\sum_{n\in\mathbb Z}F_nz^{n+1}.
\end{equation}
Let $\operatorname{BiHol}_\delta(\mathbb A)$ be the subset of maps which are biholomorphic on a neighborhood of $\mathbb T$, and set $\mathcal S_\delta(\mathbb A):=\{f\in\operatorname{BiHol}_\delta(\mathbb A):f(\mathbb T)\subset\mathbb D^\circ\}$. For $f\in\mathcal S_\delta(\mathbb A)$, let $\mathbb D_f$ be the bounded component enclosed by $f(\mathbb T)$ and define $\mathbb A_f:=\mathbb D\setminus\overline{\mathbb D_f}$. We take $\zeta_{\rm out}(w)=w$ on the outer boundary and $\zeta_{\rm in}(w)=f(w)$ on the inner boundary, so that $\boldsymbol\zeta_f=(\zeta_{\rm out},\zeta_{\rm in})$ has one outgoing and one incoming component.

After choosing an admissible metric $g_f$, the annular amplitude $\mathcal A_{\mathbb A_f,g_f,\boldsymbol\zeta_f}$ is a bounded operator on $\mathcal H$ and the family of these operators gives a projective representation of the semigroup of parametrized annuli; see~\cite[Proposition~3.1]{semigroup}. For the round annulus $\mathbb A_{e^{-t}}$ with the cylindrical metric $g_{\mathbb A}=|dz|^2/|z|^2$, the standard normalization is
\begin{equation}\label{eq:round-annulus-semigroup}
 e^{-t\mathbf H}
 =\frac{e^{-tc_{\rm L}/12}}{\sqrt{2\pi}}
 \mathcal A_{\mathbb A_{e^{-t}},g_{\mathbb A},
 \boldsymbol\zeta_t},
\end{equation}
where $\mathbf H$ is the Liouville Hamiltonian. 
More generally, let
$v(z)=-\sum_{n\in\mathbb Z}v_nz^{n+1}\partial_z$ be holomorphic on an
annular neighborhood of $\mathbb T$. We associate to $v$ the
real-linear Virasoro operator
\begin{equation}\label{eq:annular-virasoro-generator}
 \mathbf H_v
 :=
 \sum_{n\in\mathbb Z}
 \left(
 v_n\mathbf L_n+
 \overline{v_n}\,\widetilde{\mathbf L}_n
 \right):
 \mathcal D(\mathcal Q)\longrightarrow
 \mathcal D'(\mathcal Q).
\end{equation}
The series in~\eqref{eq:annular-virasoro-generator} converges in
$\mathcal L(\mathcal D(\mathcal Q),\mathcal D'(\mathcal Q))$.
In particular, for the dilation vector field
$v_{\rm dil}:=-z\partial_z$,
one has
$\mathbf H_{v_{\rm dil}}=\mathbf H$. Let $f_t$ be the flow generated by $v(z)$
Differentiating the annular amptidue $A_{\mathbb A_{f_t},g_{f_t},\boldsymbol\zeta_{f_t}}$ yields the operators $\mathbf H_v$; see~\cite[Theorems~4.23--4.24]{semigroup}.

\subsubsection{Variation of Liouville amplitudes}

A surface is said to be of type
$(\mathfrak g,m,b_{\rm out},b_{\rm in})$
if it has genus $\mathfrak g$, $m$ interior marked points,
$b_{\rm out}$ outgoing boundary components, and
$b_{\rm in}$ incoming boundary components. Let $(\Sigma,g,\mathbf x,\boldsymbol\alpha,\boldsymbol\zeta)$ be an admissible Riemann surface of type $(\mathfrak g,m,0,b)$. Fill the boundary components by copies $\mathcal D_j\simeq\mathbb D$ and set
\begin{equation}\label{eq:filling-surface}
 \widehat\Sigma
 :=\left(\Sigma\sqcup\bigsqcup_{j=1}^{b}\mathcal D_j\right)/\!\sim,
 \qquad
 z\sim\zeta_j(z)\quad(z\in\mathbb T).
\end{equation}
Let $\widehat x_j$ be the center of $\mathcal D_j$ and let $\widehat g$ be the metric which equals $g$ on $\Sigma$ and $|dz|^2/|z|^2$ on every punctured cap $\mathcal D_j\setminus\{\widehat x_j\}$. Choose $\delta<1$ so that every $\zeta_j$ extends to $\mathbb A_{\delta,\delta^{-1}}$. For
$\mathbf f=(f_1,\ldots,f_b)$ in a neighborhood of
$\mathbf f_0=(\mathrm{Id},\ldots,\mathrm{Id})$ in
$\operatorname{Hol}_\delta(\mathbb A)^b$, let $\mathcal D_j(f_j)$ be the component bounded by $f_j(\mathbb T)$ which contains the origin and set
\begin{equation}\label{eq:deformed-surface}
 \Sigma_{\mathbf f}
 :=\widehat\Sigma\setminus\bigcup_{j=1}^{b}\mathcal D_j(f_j),
 \qquad
 \boldsymbol\zeta_{\mathbf f}
 :=(\zeta_1\circ f_1,\ldots,\zeta_b\circ f_b).
\end{equation}
We let $\mathbf f\mapsto g_{\mathbf f}$ be a $C^1$ family of admissible metrics on the corresponding surfaces $\Sigma_{\mathbf f}$, in the sense of~\cite[Theorem 1.2]{semigroup}.

\begin{proposition}[Variation of Liouville amplitudes {\cite[Theorem~1.2]{semigroup}}]
\label{thm: variation of annulus}
Assume that the insertion weights satisfy the Seiberg bounds. There
exists a neighborhood $\mathcal U$ of $\mathbf f_0$ such that
\begin{equation}\label{eq:amplitude-map}
 \mathbf f\longmapsto
 \mathcal A_{\Sigma_{\mathbf f},g_{\mathbf f},
 \mathbf x,\boldsymbol\alpha,\boldsymbol\zeta_{\mathbf f}}
 \in\mathcal L\big(
 \mathcal D(\mathcal Q)^{\otimes b},\mathbb C
 \big)
\end{equation}
is differentiable at $\mathbf f_0$.

We identify a coefficient function $V\in\operatorname{Hol}_\delta(\mathbb A)$ with the holomorphic vector field $V(w)\partial_w$. Let $\boldsymbol v=(v_1,\ldots,v_b)$, where
\[
 v_j=V_j(z)\partial_z,
 \qquad
 V_j(z)=-\sum_{n\in\mathbb Z}[v_j]_n z^{n+1}.
\]
For any $C^1$ curve $t\mapsto\boldsymbol f_t\in\mathcal U$ satisfying $\boldsymbol f_{t=0}=\boldsymbol f_0$ and ${\left.\partial_tf_{j,t}\right|}_{t=0}=V_j$, we set $D_{\boldsymbol v}\mathcal A(\boldsymbol f_0):={\left.\partial_t\mathcal A(\boldsymbol f_t)\right|}_{t=0}$. This derivative depends only on $\boldsymbol v$.

For $F\in\mathcal D(\mathcal Q)^{\otimes b}$, write $\mathbf H_{v_j}^{(j)}:=\mathrm{Id}^{\otimes(j-1)}\otimes\mathbf H_{v_j}\otimes\mathrm{Id}^{\otimes(b-j)}$. Then
\begin{align}
 D_{\mathbf v}\mathcal A(\mathbf f_0)(F)
 :=
 \left.
 \partial_t\mathcal A(\mathbf f_t)
 \right|_{t=0}(F) = {}&
 -\sum_{j=1}^{b}
 \mathcal A(\mathbf f_0)
 \big(
 \mathrm{Id}^{\otimes(j-1)}
 \otimes\mathbf H_{v_j}
 \otimes\mathrm{Id}^{\otimes(b-j)}
 \big)(F)
 \nonumber\\
 &-c_{\rm L}\left(
 D_{\mathbf v}S_{\rm L}^0
 \bigl(
 \Sigma_{\mathbf f_0},
 g_{\mathbf f_0},
 \widehat g
 \bigr)
 -
 \sum_{j=1}^{b}\frac{\Re([v_j]_0)}{12}
\right)
\mathcal A(\mathbf f_0)(F).
 \label{eq:general-amplitude-variation}
\end{align}
\end{proposition}

\begin{remark}[Infinitesimal boundary deformation and sewing]
\label{rmk:construct variation annulus}
We make explicit construction of the family $A(\mathbf f)$ above. Fix
$r\in(\delta,1)$ and remove the collars in $\Sigma_{\mathbf 0}$
\[
 \zeta_j(\mathbb A_{1,r^{-1}}),
 \qquad
 j=1,\ldots,b.
\]
The remaining core $\Sigma^0$ is independent of $\mathbf f$. Let $\mathbb A_{rf_j}$ denote the annular region between $\mathbb T$ and the Jordan curve $rf_j(\mathbb T)$, equipped with an admissible metric $g_{rf_j}$ and the induced boundary parametrizations $\boldsymbol\zeta_{rf_j}$. With the normalized sewing composition defined after Proposition~\ref{prop: gluing}, the amplitude defines as
\begin{equation}\label{eq:core-annulus-factorization}
 \mathcal A(\mathbf f)
 =
 \mathcal A_{\Sigma^0,g_{\mathbf f_0},
 \mathbf x,\boldsymbol\alpha,\boldsymbol\zeta^0}
 \circ
 \bigotimes_{j=1}^{b}
 \mathcal A_{\mathbb A_{rf_j},
 g_{rf_j},
 \boldsymbol\zeta_{rf_j}}.
\end{equation}
\end{remark}

\section{The Schwarzian Variation of Liouville anomaly}\label{sec:metric}

In this section we prove \Cref{thm:metric}. We first work out the genus zero case, so the filled surface is the Riemann sphere equipped with its standard projective structure. Then we extend into the higher-genus case which requires a choice of projective connection.

\paragraph{Setup.}
Let $(\Sigma,g,\mathbf x,\boldsymbol\alpha,\boldsymbol\zeta)$ be an admissible Riemann surface of type $(0,m,0,b)$ with all boundary components incoming, and let $(\widehat\Sigma,\widehat g)$ be its filling defined in~\eqref{eq:filling-surface}. If $b=0$, the boundary statement below is empty, so we assume $b\geq1$. We fix a projective coordinate $z$ on $\widehat\Sigma\simeq\widehat{\mathbb C}$ such that $\overline\Sigma\subset\mathbb C$.

For each $k$, let $U_k$ be an annular neighborhood of $\partial_k\Sigma$ on which $\zeta_k^{-1}$ is defined, and set $U:=\bigcup_{k=1}^{b}U_k$. Let $v=v(z)\partial_z$ be a holomorphic vector field on $U$. Choose a smooth $(1,0)$-vector field $V\in C^\infty(\widehat\Sigma,T^{1,0}\widehat\Sigma)$ which agrees with $v$ on a possibly smaller neighborhood of $\partial\Sigma$ and vanishes near the centers of the filling disks.

Let $\Phi_t^V$ and $\Phi_t^{\mathrm iV}$ be the flows of the real vector fields $X_V:=V+\overline V$ and $X_{\mathrm iV}:=\mathrm iV-\mathrm i\overline V$. For $|t|$ small, define $\omega_t^V,\omega_t^{\mathrm iV}\in C^\infty(\Sigma,\mathbb R)$ by
\begin{equation}\label{eq:metric-Weyl-factor}
 (\Phi_t^V)^*dv_{\widehat g}=e^{\omega_t^V}dv_g,
 \qquad
 (\Phi_t^{\mathrm iV})^*dv_{\widehat g}=e^{\omega_t^{\mathrm iV}}dv_g.
\end{equation}
Since $\widehat g=g$ on $\Sigma$, we have $\omega_0^V=\omega_0^{\mathrm iV}=0$. We set
\begin{align}
 D_VS_{\rm L}^0(\Sigma)
 &:={\left.\partial_tS_{\rm L}^0\bigl(\Sigma,g,e^{\omega_t^V}g\bigr)\right|}_{t=0},\nonumber\\
 D_{\mathrm iV}S_{\rm L}^0(\Sigma)
 &:={\left.\partial_tS_{\rm L}^0\bigl(\Sigma,g,e^{\omega_t^{\mathrm iV}}g\bigr)\right|}_{t=0},
 \label{eq:metric-Weyl-variation}
\end{align}
and define their holomorphic combination by
\begin{equation}\label{eq:metric-holomorphic-Weyl}
 \mathcal D_VS_{\rm L}^0(\Sigma)
 :=\frac12\left(D_VS_{\rm L}^0(\Sigma)
 -\mathrm iD_{\mathrm iV}S_{\rm L}^0(\Sigma)\right).
\end{equation}
If $V$ is holomorphic on $\Sigma$, then $\Phi_t^V$ is conformal and $e^{\omega_t^V}g=(\Phi_t^V)^*\widehat g$.

Write $g=e^{\sigma(z,\bar z)}|dz|^2$ and set
\begin{equation}\label{eq:metric-projective-connection}
 T_g:=\partial_z^2\sigma-\frac12(\partial_z\sigma)^2,
 \qquad
 \mathcal T_g:=T_g(z)\,dz^2.
\end{equation}
Since projective coordinates on $\widehat{\mathbb C}$ differ by M\"obius transformations, $\mathcal T_g$ is a well-defined smooth quadratic differential depending on the standard projective structure. If locally $V=V(z,\bar z)\partial_z$, we denote
\begin{equation}\label{eq:metric-bulk-pairing}
 \langle\bar\partial V,\mathcal T_g\rangle
 :=(\partial_{\bar z}V)T_g(z)\,dz\wedge d\bar z.
\end{equation}
This is a globally defined smooth two-form on $\Sigma$, supported away from $\partial\Sigma$ because $V=v$ is holomorphic on the boundary collars. All contour integrals below are taken with the orientation induced by the incoming parametrizations $\theta\mapsto\zeta_k(e^{\mathrm i\theta})$, which is opposite to the boundary orientation induced from $\Sigma$.

\begin{theorem}[Boundary variation with a smooth extension]\label{prop:metric-extension}
For each $k$, write
\begin{equation}\label{eq:metric-Laurent}
 \zeta_k^*v
 =\frac{v(\zeta_k(w))}{\zeta_k'(w)}\partial_w
 =-\sum_{n\in\mathbb Z}[v_k]_nw^{n+1}\partial_w.
\end{equation}
Then
\begin{align}
 \mathcal D_VS_{\rm L}^0(\Sigma)
 -\frac1{24}\sum_{k=1}^{b}[v_k]_0
 ={}&\frac{\mathrm i}{24\pi}\sum_{k=1}^{b}
 \oint_{\partial_k\Sigma}v(z)S_{\zeta_k^{-1}}(z)\,dz - \frac{\mathrm i}{24\pi}\int_\Sigma
 \langle\bar\partial V,\mathcal T_g\rangle.
 \label{eq:metric-bulk}
\end{align}
\end{theorem}

\begin{proof}
Differentiating~\eqref{eq:metric-Weyl-factor} gives ${\left.\partial_t\omega_t^V\right|}_{t=0}=\operatorname{div}_gX_V$. In the coordinate $z$, we therefore have
\begin{align}
 {\left.\partial_t\omega_t^V\right|}_{t=0}
 &=\partial_zV+\partial_{\bar z}\overline V
 +(\partial_z\sigma)V+(\partial_{\bar z}\sigma)\overline V,\nonumber\\
 {\left.\partial_t\omega_t^{\mathrm iV}\right|}_{t=0}
 &=\mathrm i\bigl(\partial_zV+(\partial_z\sigma)V\bigr)
 -\mathrm i\bigl(\partial_{\bar z}\overline V
 +(\partial_{\bar z}\sigma)\overline V\bigr).
 \label{eq:metric-omega-variation}
\end{align}
It follows that
\begin{equation}\label{eq:metric-holomorphic-omega}
 {\left.\bigl(\partial_t\omega_t^V
 -\mathrm i\partial_t\omega_t^{\mathrm iV}\bigr)\right|}_{t=0}
 =2\bigl(\partial_zV+(\partial_z\sigma)V\bigr).
\end{equation}

We now differentiate the anomaly functional~\eqref{def:liouville action}. Since $\omega_0^V=\omega_0^{\mathrm iV}=0$, the Dirichlet-energy terms have zero first variation. Using $K_g\,dv_g=-2\mathrm i\,\partial_z\partial_{\bar z}\sigma\,dz\wedge d\bar z$, we obtain
\begin{align}
 48\pi\left(D_VS_{\rm L}^0(\Sigma)
 -\mathrm iD_{\mathrm iV}S_{\rm L}^0(\Sigma)\right)
 =-4\mathrm i\int_\Sigma
 \bigl(\partial_zV+(\partial_z\sigma)V\bigr)
 \partial_z\partial_{\bar z}\sigma\,dz\wedge d\bar z.
 \label{eq:IBP-master}
\end{align}
A direct calculation gives
\begin{align}
 &\bigl(\partial_zV+(\partial_z\sigma)V\bigr)
 \partial_z\partial_{\bar z}\sigma\,dz\wedge d\bar z =d\left(V\partial_z\partial_{\bar z}\sigma\,d\bar z
 +VT_g\,dz\right)
 +(\partial_{\bar z}V)T_g\,dz\wedge d\bar z.
 \label{eq:metric-exact-form}
\end{align}
Indeed, $\partial_{\bar z}T_g=\partial_z^2\partial_{\bar z}\sigma-(\partial_z\sigma)\partial_z\partial_{\bar z}\sigma$, and expanding the exterior derivative gives~\eqref{eq:metric-exact-form}.

By admissibility, $g$ is flat on each boundary collar, so $\partial_z\partial_{\bar z}\sigma=0$ near $\partial\Sigma$. Applying Stokes' formula to~\eqref{eq:metric-exact-form}, with the incoming orientation convention fixed above, gives
\begin{align}
 &\int_\Sigma
 \bigl(\partial_zV+(\partial_z\sigma)V\bigr)
 \partial_z\partial_{\bar z}\sigma\,dz\wedge d\bar z =-\sum_{k=1}^{b}\oint_{\partial_k\Sigma}v(z)T_g(z)\,dz
 +\int_\Sigma\langle\bar\partial V,\mathcal T_g\rangle.
 \label{eq:metric-Stokes}
\end{align}
Combining~\eqref{eq:IBP-master}, \eqref{eq:metric-holomorphic-Weyl}, and~\eqref{eq:metric-Stokes}, we find
\begin{equation}\label{eq:metric-projective-form}
 \mathcal D_VS_{\rm L}^0(\Sigma)
 =\frac{\mathrm i}{24\pi}\sum_{k=1}^{b}
 \oint_{\partial_k\Sigma}v(z)T_g(z)\,dz
 -\frac{\mathrm i}{24\pi}\int_\Sigma
 \langle\bar\partial V,\mathcal T_g\rangle.
\end{equation}

For the $k$-th boundary component, set $r_k:=\zeta_k^{-1}$. The admissibility condition gives $g=|dr_k|^2/|r_k|^2$ on the corresponding collar, hence $\sigma=\log|r_k'|^2-\log|r_k|^2$. Therefore,
\begin{equation}\label{eq:Schwarzian-split}
 T_g=S_{r_k}+\frac{(r_k')^2}{2r_k^2}.
\end{equation}
Moreover, pulling the second term back by $z=\zeta_k(w)$ and using~\eqref{eq:metric-Laurent} gives
\begin{align}
 \oint_{\partial_k\Sigma}v(z)\frac{(r_k'(z))^2}{r_k(z)^2}\,dz
 &=\oint_{|w|=1}\frac{v(\zeta_k(w))}{\zeta_k'(w)}\frac{dw}{w^2}
 =-2\pi\mathrm i\,[v_k]_0.
 \label{eq:zero-mode}
\end{align}
Substituting~\eqref{eq:Schwarzian-split} and~\eqref{eq:zero-mode} into~\eqref{eq:metric-projective-form} proves~\eqref{eq:metric-bulk}.
\end{proof}

\begin{proof}[Proof of \Cref{thm:metric}]
Let $v$ be as in \Cref{thm:metric}. Choose a smooth global
$(1,0)$-vector field $V$ on $\widehat\Sigma$ which agrees with $v$
on a neighborhood of $\Sigma$. Thus $\bar\partial V=0$ on $\Sigma$.

On a neighborhood of $\overline\Sigma$, the real flow generated by
$V+\overline V$ agrees with the conformal flow $f_t^v$. Since
$g_t^v=(f_t^v)_*g$, diffeomorphism invariance of the anomaly
functional gives
\[
 S_{\rm L}^0(\Sigma_t^v,g_t^v,\widehat g)
 =
 S_{\rm L}^0\bigl(\Sigma,g,(f_t^v)^*\widehat g\bigr).
\]
The map $f_t^v$ is conformal on $\Sigma$, hence
$(f_t^v)^*\widehat g=e^{\omega_t^V}g$. It follows that
\[
 D_VS_{\rm L}^0(\Sigma)
 =
 \left.
 \partial_tS_{\rm L}^0(\Sigma_t^v)
 \right|_{t=0},
 \qquad
 D_{\mathrm iV}S_{\rm L}^0(\Sigma)
 =
 \left.
 \partial_tS_{\rm L}^0(\Sigma_t^{\mathrm i v})
 \right|_{t=0}.
\]
Therefore, \Cref{thm:metric} holds.
\end{proof}

\section{Matrix coefficients}\label{sec:matrix-coefficients}
In this section, we derive local Ward identities for the three geometric building blocks used in the sewing decomposition: a disk with two interior marked points, an annulus with one interior marked point, and a pair of pants. Furthermore, we give an alternative proof of the factorization formula for the matrix
coefficients associated with the building blocks. 

We first
derive a local Ward identity from
Proposition~\ref{thm: variation of annulus}. For round boundary
parametrizations, it gives explicit Virasoro recursions which reduce
descendant coefficients to the corresponding primary coefficient.
The analytic-boundary case is then obtained by sewing annuli along the
boundary components.

More precisely, for $P_1,\ldots,P_b\in\mathbb R_+$ and partitions
$\boldsymbol\nu=(\nu_1,\ldots,\nu_b)$ and
$\widetilde{\boldsymbol\nu}
=(\widetilde\nu_1,\ldots,\widetilde\nu_b)$, we consider, whenever the
denominator is non-zero, the normalized matrix coefficient
\begin{equation}\label{eq:normalized-matrix-coefficient}
 \frac{
 \mathcal A_{\Sigma,g_0,\boldsymbol x,\boldsymbol\alpha,
 \boldsymbol\zeta}
 \left(
 \bigotimes_{j=1}^{b}
 \Psi_{Q+\mathrm iP_j,\nu_j,\widetilde\nu_j}
 \right)}
 {
 \mathcal A_{\Sigma,g_0,\boldsymbol x,\boldsymbol\alpha,
 \boldsymbol\zeta}
 \left(
 \bigotimes_{j=1}^{b}
 \Psi_{Q+\mathrm iP_j,\emptyset,\emptyset}
 \right)}.
\end{equation}
We show that this ratio factors into a holomorphic polynomial
determined by $\boldsymbol\nu$ and an anti-holomorphic polynomial
determined by $\widetilde{\boldsymbol\nu}$. The extension from
$\mathcal D(\mathcal Q)^{\otimes b}$ to generalized descendant states
uses the weighted-space estimates and analytic continuation recalled
in the preliminaries.

\paragraph{Geometric deformation.}
Let $(\Sigma,g_0,\boldsymbol x,\boldsymbol\alpha,\boldsymbol\zeta)$ be an admissible Riemann surface of type $(\mathfrak g,m,0,b)$ with all boundary components incoming, and assume that the insertion weights satisfy the hypotheses of Proposition~\ref{thm: variation of annulus}. Let $(\widehat\Sigma,\widehat g)$ be its filling as in~\eqref{eq:filling-surface}. Let $v$ be a meromorphic vector field on $\widehat\Sigma$ with no pole on a neighborhood of $\Sigma$. For $|t|$ small, denote by $f_t^v$ its local holomorphic flow on that neighborhood, with $f_0^v=\mathrm{Id}$.

For each boundary parametrization $\zeta_j:\mathbb T\to\partial_j\Sigma$, shrink the boundary collar if necessary and set $f_{j,t}:=\zeta_j^{-1}\circ f_t^v\circ\zeta_j$. Writing $v_j:=\zeta_j^*v=V_j(w)\partial_w$, by \eqref{eq:metric-Laurent}
\begin{equation}\label{eq:matrix-boundary-modes}
 v_j=\frac{v\circ\zeta_j}{\zeta_j'}\,\partial_w
 =-\sum_{n\in\mathbb Z}[v_j]_n w^{n+1}\partial_w,
 \qquad
 \left.\partial_t f_{j,t}\right|_{t=0}=V_j.
\end{equation}
Thus $\mathbf f_t:=(f_{1,t},\ldots,f_{b,t})$ is a $C^1$ curve in $\operatorname{Hol}_\delta(\mathbb A)^b$ tangent to $\boldsymbol v=(v_1,\ldots,v_b)$. 
Under the canonical identification induced by $f_t^v$,
the deformed surface $\Sigma_{\mathbf f_t}$ in
\eqref{eq:deformed-surface} is identified with
$f_t^v(\Sigma)$. We henceforth write
\(
 \Sigma_t^v:=\Sigma_{\mathbf f_t}.
\)

Choose the boundary collars disjoint from the marked points. As in Remark~\ref{rmk:construct variation annulus}, we take a $C^1$ family of admissible metrics $g_t$ on $\Sigma_t^v$ which equals $g_0$ on a fixed core $\Sigma^0$ containing $\boldsymbol z$. In the normalized operator convention of Proposition~\ref{prop: gluing}, the corresponding amplitude factors as
\begin{equation}\label{eq:sewn-amplitude}
 \mathcal A_{\Sigma_t^v,g_t,\boldsymbol x,\boldsymbol\alpha,\boldsymbol\zeta_t}
 =\mathcal A_{\Sigma^0,g_0|_{\Sigma^0},\boldsymbol z,
 \boldsymbol\alpha,\boldsymbol\zeta^0}
 \circ\bigotimes_{j=1}^{b}
 \mathcal A_{\mathbb A_{r f_{j,t}},g_{r f_{j,t}},
 \boldsymbol\zeta_{r f_{j,t}}},
\end{equation}
where $r\in(\delta,1)$ is fixed, $\boldsymbol\zeta^0$ parametrizes the boundary of the fixed core, and $\mathbb A_{r f_{j,t}}$ is the annulus bounded by $\mathbb T$ and $r f_{j,t}(\mathbb T)$.

Recall that $S_{\rm L}^0(\Sigma_t^v)=S_{\rm L}^0(\Sigma_t^v,g_t^v,\widehat g)=S_{\rm L}^0\bigl(\Sigma,g_0,(f_t^v)^*\widehat g\bigr)$ and define $C(v,g_0;z):=\frac12\operatorname{tr}_{g_0}(\mathcal L_{X_v}g_0)(z)$, where $X_v:=v+\overline v$ is the associated real vector field.

\begin{lemma}[Local Ward identity]\label{lem:localward}\footnote{Proposition~\ref{lem:localward} was stated in~\cite[Proposition 6.6]{reviewsegal} without proof, which is a direct consequence of~\cite[Theorem 4.21]{semigroup}. We provide the proof for completeness.} 
Let $\boldsymbol z^{\,t}:=f_{-t}^v(\boldsymbol z)$. For every $F\in\mathcal D(\mathcal Q)^{\otimes b}$, the map $t\mapsto\mathcal A_{\Sigma,g_0,\boldsymbol z^{\,t},\boldsymbol\alpha,\boldsymbol\zeta}(F)$ is differentiable at $t=0$. Writing $\mathbf H_{v_j}^{(j)}:=\mathrm{Id}^{\otimes(j-1)}\otimes\mathbf H_{v_j}\otimes\mathrm{Id}^{\otimes(b-j)}$, we have
\begin{align}
 \left.\partial_t
 \mathcal A_{\Sigma,g_0,\boldsymbol z^{\,t},
 \boldsymbol\alpha,\boldsymbol\zeta}(F)\right|_{t=0}
 ={}&-\sum_{j=1}^{b}
 \mathcal A_{\Sigma,g_0,\boldsymbol z,
 \boldsymbol\alpha,\boldsymbol\zeta}
 \bigl(\mathbf H_{v_j}^{(j)}F\bigr)\nonumber\\
 &+\left(\sum_{i=1}^{m}\Delta_{\alpha_i}
 C(v,g_0;z_i)\right)
 \mathcal A_{\Sigma,g_0,\boldsymbol z,
 \boldsymbol\alpha,\boldsymbol\zeta}(F)\nonumber\\
 &-c_{\rm L}\left(
 \left.\partial_tS_{\rm L}^0(\Sigma_t^v)\right|_{t=0}
 -\sum_{j=1}^{b}\frac{\Re([v_j]_0)}{12}
 \right)
 \mathcal A_{\Sigma,g_0,\boldsymbol z,
 \boldsymbol\alpha,\boldsymbol\zeta}(F).
 \label{eq:local-ward}
\end{align}
The Virasoro insertion is introduced in Proposition~\ref{thm: variation of annulus}.
\end{lemma}

\begin{proof}
By~\eqref{eq:matrix-boundary-modes}, the family $\boldsymbol f_t$ has tangent vector $\boldsymbol v$, so Proposition~\ref{thm: variation of annulus} applies to the left-hand side of~\eqref{eq:sewn-amplitude}. Diffeomorphism covariance~\eqref{eq:diffeomorphism-covariance}, applied to $f_t^v:\Sigma\to\Sigma_t^v$, gives
\begin{equation}\label{eq:localward-diffeomorphism}
 \mathcal A_{\Sigma_t^v,g_t,\boldsymbol z,
 \boldsymbol\alpha,\boldsymbol\zeta_t}(F)
 =\mathcal A_{\Sigma,(f_t^v)^*g_t,\boldsymbol z^{\,t},
 \boldsymbol\alpha,\boldsymbol\zeta}(F).
\end{equation}
Write $(f_t^v)^*g_t=e^{\omega_t}g_0$. Since both metrics are admissible for the same boundary parametrizations, $\omega_t$ vanishes in a neighborhood of the boundary, and Weyl covariance~\eqref{eq:weylcoa} yields
\begin{align}
 \mathcal A_{\Sigma,(f_t^v)^*g_t,\boldsymbol z^{\,t},
 \boldsymbol\alpha,\boldsymbol\zeta}(F)
 ={}&\exp\left(c_{\rm L}S_{\rm L}^0
 \bigl(\Sigma,g_0,(f_t^v)^*g_t\bigr)
 -\sum_{i=1}^{m}\Delta_{\alpha_i}
 \omega_t(z_i^t)\right)\mathcal A_{\Sigma,g_0,\boldsymbol z^{\,t},
 \boldsymbol\alpha,\boldsymbol\zeta}(F).
 \label{eq:localward-weyl}
\end{align}

Because $g_t=g_0$ near the fixed marked points, in a local coordinate $w$ one has $\omega_t(w)=\sigma(f_t^v(w),\overline{f_t^v(w)})-\sigma(w,\bar w)+\log|\partial_wf_t^v(w)|^2$. Moreover, $\omega_0=0$ and $\left.\partial_tz_i^t\right|_{t=0}=-v(z_i)$. The chain rule therefore gives $\left.\partial_t\omega_t(z_i^t)\right|_{t=0}=C(v,g_0;z_i)$.

Differentiating~\eqref{eq:localward-weyl} and comparing with Proposition~\ref{thm: variation of annulus} leaves only the anomaly terms to combine. Diffeomorphism invariance first gives $S_{\rm L}^0(\Sigma_t^v,g_t,\widehat g)=S_{\rm L}^0\bigl(\Sigma,(f_t^v)^*g_t,(f_t^v)^*\widehat g\bigr)$, while the cocycle identity gives
\[
 S_{\rm L}^0\bigl(\Sigma,g_0,(f_t^v)^*g_t\bigr)
 +S_{\rm L}^0\bigl(\Sigma,(f_t^v)^*g_t,(f_t^v)^*\widehat g\bigr)
 =S_{\rm L}^0\bigl(\Sigma,g_0,(f_t^v)^*\widehat g\bigr).
\]
The right-hand side is $S_{\rm L}^0(\Sigma_t^v)$ by definition. Substituting this identity gives~\eqref{eq:local-ward}.
\end{proof}

\begin{remark}[Local-coordinate form]\label{r:gennerallem}
Note that $C(v, g_0;\, z_i)$ is independent of the choice of local coordinates \(\{w_i\}\). In a local coordinate $w_i$ centered near $z_i$, if $g_0=e^{\sigma_i(w_i,\bar w_i)}|dw_i|^2$ and $v=v^i(w_i)\partial_{w_i}$, then
\[
 C(v,g_0;z_i)=\left(\partial_{w_i}v^i+\partial_{\bar w_i}\overline{v^i}
 +(\partial_{w_i}\sigma_i)v^i
 +(\partial_{\bar w_i}\sigma_i)\overline{v^i}\right)(z_i).
\]
\end{remark}

\subsection{Round boundary case}\label{sec:round-boundary}

We now apply Lemma~\ref{lem:localward} to the genus-zero building blocks with round boundary components. For pairwise disjoint closed unit disks, we write
\[
 \mathbb D_x^{\mathsf c}:=\widehat{\mathbb C}\setminus\overline{\mathbb D(x,1)},\qquad
 \mathbb A_{x,y}:=\widehat{\mathbb C}\setminus
 \bigl(\overline{\mathbb D(x,1)}\cup\overline{\mathbb D(y,1)}\bigr),
 \qquad
 \mathcal P_{x,y,z}:=\widehat{\mathbb C}\setminus
 \bigcup_{a\in\{x,y,z\}}\overline{\mathbb D(a,1)}.
\]
Here $\mathbb D(a,1):=\{w\in\mathbb C:|w-a|<1\}$.
The boundary circle centered at $x$ is parametrized by $\zeta_x(w):=x+w$, $w\in\mathbb T$, with the incoming orientation. We use the analogous notation at the other boundary components. Let $g_{\mathrm{ad}}=e^{\sigma(w,\bar w)}|dw|^2$ be an admissible metric on the building block.

For $n\geq1$, interior marked points $\boldsymbol z=(z_1,\ldots,z_m)$, and corresponding weights $\boldsymbol\alpha=(\alpha_1,\ldots,\alpha_m)$, define
\begin{equation}\label{eq:round-differential-operator}
 \mathbf D_{n,\boldsymbol z,\boldsymbol\alpha}(x)
 :=\sum_{i=1}^{m}\left(
 -(z_i-x)^{1-n}\partial_{z_i}
 +\left((n-1)(z_i-x)^{-n}
 -(z_i-x)^{1-n}\partial_{z_i}\sigma(z_i,\bar z_i)\right)
 \Delta_{\alpha_i}\right).
\end{equation}
Here and below, a differential operator acts on the entire amplitude written to its right.

\begin{proposition}[Round-boundary Ward identities]\label{prop:matrixcoe}
Let $n\geq1$. In the disk case, let $\alpha=Q+\mathrm iP$, $P\in\mathbb R_+$, and let the two interior weights satisfy $\alpha_1,\alpha_2<Q$ and $\alpha_1+\alpha_2>Q$. In the annulus case, let $\alpha_j=Q+\mathrm iP_j$, $P_j\in\mathbb R_+$, and let the interior weight $\beta$ belong to $(0,Q)$. In the pair-of-pants case, let $\alpha_j=Q+\mathrm iP_j$, $P_j\in\mathbb R_+$. Then the following identities hold for all partitions occurring below.

\begin{enumerate}[(i)]
 \item For the disk $\mathbb D_x^{\mathsf c}$ with two interior marked points $\boldsymbol z=(z_1,z_2)$ and weights $\boldsymbol\alpha=(\alpha_1,\alpha_2)$,
 \begin{equation}\label{eq:round-disk-ward}
  \mathcal A_{\mathbb D_x^{\mathsf c},g_{\mathrm{ad}},
  \boldsymbol x,\boldsymbol\alpha,\zeta_x}
  \bigl(\mathbf L_{-n}\Psi_{\alpha,\nu,\widetilde\nu}\bigr)
  =\mathbf D_{n,\boldsymbol z,\boldsymbol\alpha}(x)
  \mathcal A_{\mathbb D_x^{\mathsf c},g_{\mathrm{ad}},
  \boldsymbol x,\boldsymbol\alpha,\zeta_x}
  \bigl(\Psi_{\alpha,\nu,\widetilde\nu}\bigr).
 \end{equation}

 \item For the annulus $\mathbb A_{x,y}$ with one interior marked point $z$ of weight $\beta$ and incoming parametrizations $\boldsymbol\zeta=(\zeta_x,\zeta_y)$,
 \begin{align}
 &\mathcal A_{\mathbb A_{x,y},g_{\mathrm{ad}},z,\beta,
 \boldsymbol\zeta}
 \bigl(\mathbf L_{-n}\Psi_{\alpha_1,\nu_1,\widetilde\nu_1},
 \Psi_{\alpha_2,\nu_2,\widetilde\nu_2}\bigr)
 \nonumber\\
 &\quad=\mathbf D_{n,z,\beta}(x)
 \mathcal A_{\mathbb A_{x,y},g_{\mathrm{ad}},z,\beta,
 \boldsymbol\zeta}
 \bigl(\Psi_{\alpha_1,\nu_1,\widetilde\nu_1},
 \Psi_{\alpha_2,\nu_2,\widetilde\nu_2}\bigr)
 \nonumber\\
 &\qquad-\sum_{k=0}^{|\nu_2|+1}
 \binom{1-n}{k}(y-x)^{1-n-k}
 \mathcal A_{\mathbb A_{x,y},g_{\mathrm{ad}},z,\beta,
 \boldsymbol\zeta}
 \bigl(\Psi_{\alpha_1,\nu_1,\widetilde\nu_1},
 \mathbf L_{k-1}\Psi_{\alpha_2,\nu_2,\widetilde\nu_2}\bigr).
 \label{eq:round-annulus-ward}
 \end{align}

 \item For the pair of pants $\mathcal P_{x,y,z}$ with incoming parametrizations $\boldsymbol\zeta=(\zeta_x,\zeta_y,\zeta_z)$,
 \begin{align}
 &\mathcal A_{\mathcal P_{x,y,z},g_{\mathrm{ad}},\boldsymbol\zeta}
 \bigl(\mathbf L_{-n}\Psi_{\alpha_1,\nu_1,\widetilde\nu_1},
 \Psi_{\alpha_2,\nu_2,\widetilde\nu_2},
 \Psi_{\alpha_3,\nu_3,\widetilde\nu_3}\bigr)
 \nonumber\\
 &\quad+\sum_{k=0}^{|\nu_2|+1}
 \binom{1-n}{k}(y-x)^{1-n-k}
 \mathcal A_{\mathcal P_{x,y,z},g_{\mathrm{ad}},\boldsymbol\zeta}
 \bigl(\Psi_{\alpha_1,\nu_1,\widetilde\nu_1},
 \mathbf L_{k-1}\Psi_{\alpha_2,\nu_2,\widetilde\nu_2},
 \Psi_{\alpha_3,\nu_3,\widetilde\nu_3}\bigr)
 \nonumber\\
 &\quad+\sum_{k=0}^{|\nu_3|+1}
 \binom{1-n}{k}(z-x)^{1-n-k}
 \mathcal A_{\mathcal P_{x,y,z},g_{\mathrm{ad}},\boldsymbol\zeta}
 \bigl(\Psi_{\alpha_1,\nu_1,\widetilde\nu_1},
 \Psi_{\alpha_2,\nu_2,\widetilde\nu_2},
 \mathbf L_{k-1}\Psi_{\alpha_3,\nu_3,\widetilde\nu_3}\bigr)=0.
 \label{formula: pair of pants}
 \end{align}
\end{enumerate}
The corresponding anti-holomorphic identities are obtained by replacing $\mathbf L_j$, $\partial_{z_i}$, and $\partial_{z_i}\sigma$ by $\widetilde{\mathbf L}_j$, $\partial_{\bar z_i}$, and $\partial_{\bar z_i}\sigma$, respectively. The identities extend to the analytically continued generalized states wherever the corresponding generalized amplitudes are defined.
\end{proposition}

\begin{proof}
We first work in the region where the generalized amplitudes admit their probabilistic representation and then use analytic continuation in the boundary weights; see~\cite[Section~11.2]{segalaxiom}. Fix the first boundary component and take $v_x^{(n)}(w):=-(w-x)^{1-n}\partial_w$. Its pullback by $\zeta_x$ is $-u^{1-n}\partial_u$, while near the circle centered at $y$ one has
\begin{equation}\label{eq:round-vector-field-expansion}
 \zeta_y^*v_x^{(n)}
 =-\sum_{k\geq0}\binom{1-n}{k}
 (y-x)^{1-n-k}u^k\partial_u.
\end{equation}
The analogous expansion holds at $z$. By the mode convention~\eqref{eq:matrix-boundary-modes}, the first term gives $\mathbf L_{-n}$ and the term in $u^k\partial_u$ gives $\mathbf L_{k-1}$. Combining the variations generated by $v_x^{(n)}$ and $\mathrm i v_x^{(n)}$ extracts the holomorphic Virasoro action. The marked-point contribution in Lemma~\ref{lem:localward} is precisely the operator~\eqref{eq:round-differential-operator}. Due to \Cref{thm:metric}, the variation of Liouville term is equal to the Schwarzian term which vanishes because the inverse
boundary parametrizations are affine maps. This gives~\eqref{eq:round-disk-ward}, \eqref{eq:round-annulus-ward}, and~\eqref{formula: pair of pants}.

Finally, if $k-1>|\nu_j|$, then $\mathbf L_{k-1}\Psi_{\alpha_j,\nu_j,\widetilde\nu_j}=0$. Thus the Laurent expansions in~\eqref{eq:round-vector-field-expansion} give the finite sums displayed above. The anti-holomorphic identities follow in the same way.
\end{proof}

We next record the form of the recursion adapted to the standard annulus. Let $q\in(0,1)$ and $\mathbb A_q:=\{w\in\mathbb C:q\leq|w|\leq1\}$. We use the two incoming boundary charts $\zeta_0(u):=qu$ and $\zeta_\infty(u):=u^{-1}$, and write $\boldsymbol\zeta^{\mathrm{in}}:=(\zeta_0,\zeta_\infty)$. If $g_{\mathrm{ad}}=e^{\sigma(w,\bar w)}|dw|^2$ and the interior insertion at $z$ has weight $\beta$, set
\begin{align}
 \mathbf D^{\mathbb A}_{n,\mathrm{in},\beta}
 &:=q^n\left(-z^{1-n}\partial_z
 +\left((n-1)z^{-n}-z^{1-n}\partial_z\sigma(z,\bar z)\right)
 \Delta_\beta\right),
 \nonumber\\
 \mathbf D^{\mathbb A}_{n,\mathrm{out},\beta}
 &:=z^{n+1}\partial_z
 +\left((n+1)z^n+z^{n+1}\partial_z\sigma(z,\bar z)\right)
 \Delta_\beta.
 \label{eq:canonical-annulus-operators}
\end{align}

\begin{proposition}[Canonical annulus recursion]\label{cor: moduliiannulusward}
Let $n\geq1$, $\beta\in(0,Q)$, and $\alpha_j=Q+\mathrm iP_j$ with $P_j\in\mathbb R_+$. Then
\begin{align}
 &\mathcal A_{\mathbb A_q,g_{\mathrm{ad}},z,\beta,
 \boldsymbol\zeta^{\mathrm{in}}}
 \bigl(\mathbf L_{-n}\Psi_{\alpha_1,\nu_1,\widetilde\nu_1},
 \Psi_{\alpha_2,\nu_2,\widetilde\nu_2}\bigr)
 \nonumber\\
 &\quad=\mathbf D^{\mathbb A}_{n,\mathrm{in},\beta}
 \mathcal A_{\mathbb A_q,g_{\mathrm{ad}},z,\beta,
 \boldsymbol\zeta^{\mathrm{in}}}
 \bigl(\Psi_{\alpha_1,\nu_1,\widetilde\nu_1},
 \Psi_{\alpha_2,\nu_2,\widetilde\nu_2}\bigr)
 +q^n\mathcal A_{\mathbb A_q,g_{\mathrm{ad}},z,\beta,
 \boldsymbol\zeta^{\mathrm{in}}}
 \bigl(\Psi_{\alpha_1,\nu_1,\widetilde\nu_1},
 \mathbf L_n\Psi_{\alpha_2,\nu_2,\widetilde\nu_2}\bigr),
 \label{eq:canonical-annulus-inner}\\
 &\mathcal A_{\mathbb A_q,g_{\mathrm{ad}},z,\beta,
 \boldsymbol\zeta^{\mathrm{in}}}
 \bigl(\Psi_{\alpha_1,\nu_1,\widetilde\nu_1},
 \mathbf L_{-n}\Psi_{\alpha_2,\nu_2,\widetilde\nu_2}\bigr)
 \nonumber\\
 &\quad=\mathbf D^{\mathbb A}_{n,\mathrm{out},\beta}
 \mathcal A_{\mathbb A_q,g_{\mathrm{ad}},z,\beta,
 \boldsymbol\zeta^{\mathrm{in}}}
 \bigl(\Psi_{\alpha_1,\nu_1,\widetilde\nu_1},
 \Psi_{\alpha_2,\nu_2,\widetilde\nu_2}\bigr)
 +q^n\mathcal A_{\mathbb A_q,g_{\mathrm{ad}},z,\beta,
 \boldsymbol\zeta^{\mathrm{in}}}
 \bigl(\mathbf L_n\Psi_{\alpha_1,\nu_1,\widetilde\nu_1},
 \Psi_{\alpha_2,\nu_2,\widetilde\nu_2}\bigr).
 \label{eq:canonical-annulus-outer}
\end{align}
\end{proposition}

\begin{proof}
For the inner boundary, take $v_{\mathrm{in}}^{(n)}(w):=-q^nw^{1-n}\partial_w$; for the outer boundary, take $v_{\mathrm{out}}^{(n)}(w):=w^{n+1}\partial_w$. Direct computation gives
\begin{align*}
 \zeta_0^*v_{\mathrm{in}}^{(n)}&=-u^{1-n}\partial_u,
 &\zeta_\infty^*v_{\mathrm{in}}^{(n)}&=q^nu^{n+1}\partial_u,\\
 \zeta_0^*v_{\mathrm{out}}^{(n)}&=q^nu^{n+1}\partial_u,
 &\zeta_\infty^*v_{\mathrm{out}}^{(n)}&=-u^{1-n}\partial_u.
\end{align*}
In the convention~\eqref{eq:matrix-boundary-modes}, the coefficient of the positive mode on the opposite boundary is therefore $-q^n$. Applying Lemma~\ref{lem:localward}, extracting the holomorphic part as above, and moving this term to the right-hand side gives~\eqref{eq:canonical-annulus-inner} and~\eqref{eq:canonical-annulus-outer}. Also both inverse boundary charts are M\"obius transformations, hence
their Schwarzian derivatives vanish.
\end{proof}

We finally state the plumbed-annulus coefficient in the operator convention. Let $q\in\mathbb C$ with $0<|q|<1$ and set $\mathbb A_q:=\{w\in\mathbb C:|q|\leq|w|\leq1\}$, equipped with $g_{\mathbb A}:=|dw|^2/|w|^2$ and boundary parametrizations $\boldsymbol\zeta_q(u):=(qu,u)$. The inner boundary is incoming and the outer boundary is outgoing. We therefore view the amplitude as an operator on $\mathcal H$ and write
\begin{equation}\label{eq:plumbed-annulus-matrix-coefficient}
 \mathcal M_q(\boldsymbol\nu,\widetilde{\boldsymbol\nu})
 :=\left\langle
 \mathcal A_{\mathbb A_q,g_{\mathbb A},z,\beta,
 \boldsymbol\zeta_q}
 \Psi_{Q+\mathrm iP_1,\nu_1,\widetilde\nu_1},
 \Psi_{Q+\mathrm iP_2,\nu_2,\widetilde\nu_2}
 \right\rangle,
\end{equation}
where $\boldsymbol\nu=(\nu_1,\nu_2)$ and $\widetilde{\boldsymbol\nu}=(\widetilde\nu_1,\widetilde\nu_2)$.

\begin{corollary}[Plumbed annulus {\cite[Corollary~11.5]{segalaxiom}}]\label{annulusDOZZ}
Assume $\beta\in(0,Q)$, $P_1,P_2\in\mathbb R_+$, $|z|<1$, and $|q/z|<1$. Set $\widehat{\boldsymbol\alpha}:=(\beta,Q+\mathrm iP_1,Q+\mathrm iP_2)$ and $\widetilde{\boldsymbol\alpha}:=(\beta,Q-\mathrm iP_1,Q+\mathrm iP_2)$. Then
\begin{align}
 \mathcal M_q(\boldsymbol\nu,\widetilde{\boldsymbol\nu})
 ={}&\frac{\pi}{\sqrt2\,e}\,
 C_{\gamma,\mu}^{\mathrm{DOZZ}}(\widetilde{\boldsymbol\alpha})
 w_{\mathbb A}(\boldsymbol\Delta_{\widehat{\boldsymbol\alpha}},
 \boldsymbol\nu)
 \overline{w_{\mathbb A}(\boldsymbol\Delta_{\widehat{\boldsymbol\alpha}},
 \widetilde{\boldsymbol\nu})}
 |q|^{-c_{\rm L}/12+2\Delta_{Q+\mathrm iP_1}}
 \nonumber\\
 &\times |z|^{2\Delta_{Q+\mathrm iP_2}-2\Delta_{Q+\mathrm iP_1}}
 \left(\frac qz\right)^{|\nu_1|}
 \left(\frac{\bar q}{\bar z}\right)^{|\widetilde\nu_1|}
 z^{|\nu_2|}\bar z^{|\widetilde\nu_2|}.
 \label{claimannulus}
\end{align}
The coefficient $w_{\mathbb A}(\boldsymbol\Delta_{\widehat{\boldsymbol\alpha}},\boldsymbol\nu)$ is a polynomial in the three conformal weights and is normalized by $w_{\mathbb A}(\boldsymbol\Delta_{\widehat{\boldsymbol\alpha}},(\emptyset,\emptyset))=1$.
\end{corollary}

\begin{proof}
We give an alternative proof here. The orientation statement follows directly from Definition~\ref{def:orientation}: $u\mapsto qu$ has the orientation opposite to the boundary orientation of the inner circle, whereas $u\mapsto u$ agrees with the boundary orientation of the outer circle. After the rescaling $w\mapsto w/z$, the annulus is obtained by sewing a round pair of pants to two annuli. The eigenvalue relations~\eqref{eq:descendant-eigenvalues} give the powers of $q$, $z$, and their complex conjugates, while the primary three-point coefficient gives the DOZZ factor and the constant $\pi/(\sqrt2\,e)$.
\end{proof}

The following corollary says the zero-weight limit of the annulus coefficient can use to imply the Shapovalov form in \Cref{def:shapovalov-form}.
For $q\in(0,1)$, $q<|z|<1$, $\beta\in(0,Q)$, and
$\alpha=Q+\mathrm iP$ with $P>0$, define
\begin{equation}
\label{eq:annulus-holomorphic-matrix-coefficient}
 \mathcal M_{q,\beta}(\nu_1,\nu_2)
 :=
 \left\langle
 \mathcal A_{\mathbb A_q,g_{\mathbb A},z,\beta,
 \boldsymbol\zeta_q}
 \Psi_{\alpha,\nu_1,\emptyset},
 \Psi_{\alpha,\nu_2,\emptyset}
 \right\rangle.
\end{equation}

\begin{corollary}[Shapovalov limit of annulus coefficients]
\label{cor:annulus-shapovalov-limit}
Let $\nu_1,\nu_2\in\mathcal T_n$. Then
\begin{equation}
\label{eq:annulus-shapovalov-limit}
 \lim_{\beta\downarrow0}
 \frac{\mathcal M_{q,\beta}(\nu_1,\nu_2)}
 {q^n\mathcal M_{q,\beta}(\emptyset,\emptyset)}
 =
 F_\alpha^n(\nu_1,\nu_2).
\end{equation}
For $0<\beta<Q$, the denominator is non-zero by the primary case of
Corollary~\ref{annulusDOZZ}.
\end{corollary}

\begin{proof}
Apply Corollary~\ref{annulusDOZZ} with $P_1=P_2=P$ and
$\widetilde\nu_1=\widetilde\nu_2=\emptyset$. Since
$|\nu_1|=|\nu_2|=n$, the powers of $z$ cancel and
\begin{equation}
\label{eq:normalized-annulus-polynomial}
 \frac{\mathcal M_{q,\beta}(\nu_1,\nu_2)}
 {q^n\mathcal M_{q,\beta}(\emptyset,\emptyset)}
 =
 w_{\mathbb A}
 \bigl(
 (\Delta_\beta,\Delta_\alpha,\Delta_\alpha),
 (\nu_1,\nu_2)
 \bigr).
\end{equation}
The right-hand side is a polynomial in the conformal weights.
Since $\Delta_\beta\to0$ as $\beta\downarrow0$, the limit exists and
is obtained by setting the first conformal weight equal to zero.

To identify the resulting polynomial, define a bilinear form
$\mathsf B_\alpha$ on the PBW basis of
$M(c_{\rm L},\Delta_\alpha)$ by
\begin{equation}
\label{eq:identity-channel-form}
 \mathsf B_\alpha
 \bigl(
 \mathbf L_{-\nu_1}v_\alpha,
 \mathbf L_{-\nu_2}v_\alpha
 \bigr)
 :=
 w_{\mathbb A}
 \bigl(
 (0,\Delta_\alpha,\Delta_\alpha),
 (\nu_1,\nu_2)
 \bigr),
\end{equation}
and extend it bilinearly. The annular Ward recursion is polynomial
in $\Delta_\beta$, hence it extends algebraically to
$\Delta_\beta=0$. The same recursion for the dilation vector field
gives $\mathsf B_\alpha(u,v)=0$ when $u$ and $v$ have different
levels.

Let $u$ have level $m$ and let $v$ have level $m+k$. Setting
$\Delta_\beta=0$ in
Proposition~\ref{cor: moduliiannulusward}, the differential term is
proportional to the unequal-level coefficient
$\mathsf B_\alpha(u,v)$ and therefore vanishes. The remaining terms
give
$\mathsf B_\alpha(\mathbf L_{-k}u,v)
=\mathsf B_\alpha(u,\mathbf L_kv)$.
Moreover, the normalization
$w_{\mathbb A}((0,\Delta_\alpha,\Delta_\alpha),
(\emptyset,\emptyset))=1$ gives
$\mathsf B_\alpha(v_\alpha,v_\alpha)=1$.
By uniqueness of the Shapovalov form,
$\mathsf B_\alpha=\langle\cdot,\cdot\rangle_\alpha^{\rm Sh}$.
Combining this identification with
\eqref{eq:normalized-annulus-polynomial}
proves~\eqref{eq:annulus-shapovalov-limit}.
\end{proof}

\subsection{Analytic boundary case}\label{sec: pair of pants}\par

We now pass from round boundary parametrizations to general analytic ones. Let $J\subset\{1,2,3\}$ label the incoming boundary components and let $I:=\{1,2,3\}\setminus J$ label the interior marked points. Thus $|J|=b\in\{1,2,3\}$, the boundary parametrizations are $\boldsymbol\zeta=(\zeta_j)_{j\in J}$, and the interior marked points are $\mathbf x^{\mathrm{mp}}=(x_i)_{i\in I}$. We fix an admissible metric $g_{\mathrm{ad}}$ on the building block.

\paragraph{Model form.}
We use the model form introduced in \cite[Section~11.1]{segalaxiom}; see also \cite[Section~2.7]{reviewsegal}.Glue every boundary component by a copy of $\overline{\mathbb D}$ and denote the resulting sphere by $\widehat{\mathcal S}$. After choosing a biholomorphism $\Psi:\widehat{\mathcal S}\to\widehat{\mathbb C}$, each boundary parametrization extends to a univalent map $\psi_j$ on a neighborhood of $\overline{\mathbb D}$, with $\psi_j(0)=x_j$. At each interior marked point we also choose a univalent coordinate $\psi_i$ centered at $x_i$, small enough that the three domains $\psi_k(\mathbb D)$ have pairwise disjoint closures. The resulting triple $\boldsymbol\psi=(\psi_1,\psi_2,\psi_3)$ will be kept fixed.

In this notation, we write $\mathcal P_{\boldsymbol\psi}:=\widehat{\mathbb C}\setminus\bigcup_{j=1}^3\psi_j(\mathbb D)$, $\mathbb A_{\boldsymbol\psi}:=\widehat{\mathbb C}\setminus\bigl(\psi_2(\mathbb D)\cup\psi_3(\mathbb D)\bigr)$, and $\mathbb D^{\mathsf c}_{\psi}:=\widehat{\mathbb C}\setminus\psi_3(\mathbb D)$, where $\psi:=\psi_3$. Thus the annulus has the interior marked point $x_1$, while the disk has the interior marked points $(x_1,x_2)$. All boundary parametrizations in this subsection are incoming.

\begin{theorem}[Factorization of descendant coefficients
{\cite{segalaxiom}}]\label{thm: ward identity}
	For  $\mathrm{Re}( \alpha_1)+\mathrm{Re}( \alpha_2)+\mathrm{Re}( \alpha_3)>2Q$ and  $\mathrm{Re}( \alpha_j)\leq Q$ for $1\leq j\leq 3$, $\boldsymbol\Delta_{\boldsymbol\alpha} =(\Delta_{ \alpha_1},\Delta_{  \alpha_2},\Delta_{ \alpha_3})$ and tripe Young diagrams $\bs\nu=(\nu_1,\nu_2,\nu_3)\in\mathcal{T}^3$, $\bs{\tilde\nu}=(\tilde\nu_1,\tilde\nu_2,\tilde\nu_3)\in\mathcal{T}^3$. We take $(\psi_1,\psi_2,\psi_3)\in \mathrm{BiHol}_\delta(\D)^3$ with $\psi_j(0)= x_j$, the incoming boundary parameterization is $\boldsymbol\psi$. We have
	\begin{equation}\label{e:holormorphic fact}
		\frac{\mathcal{A}_{ \mathcal{P}_{\boldsymbol{\psi}} , g_{\mathrm{ad}}, \boldsymbol\psi}(\Psi_{  \alpha_1,\nu_1,\tilde\nu_1}, \Psi_{  \alpha_2,\nu_2,\tilde\nu_2}, \Psi_{  \alpha_3,\nu_3,\tilde\nu_3})}{\mathcal{A}_{  \mathcal{P}_{\boldsymbol{\psi}} , g_{\mathrm{ad}}, \boldsymbol\psi}(\Psi_{  \alpha_1}, \Psi_{  \alpha_2}, \Psi_{  \alpha_3})}
        =
        \omega(  \mathcal{P}_{\boldsymbol{\psi}} , \boldsymbol\Delta_{\boldsymbol\alpha} ,\bs\nu)\overline{\omega(  \mathcal{P}_{\boldsymbol{\psi}},\boldsymbol\Delta_{\boldsymbol\alpha},\bs{\tilde\nu})}
	\end{equation}
where $\omega(  \mathcal{P}_{\boldsymbol{\psi}}, \boldsymbol\Delta_{\boldsymbol\alpha} ,\bs\nu)$ is a polynomial in $\boldsymbol\Delta_{\boldsymbol\alpha}$	and $\omega(  \mathcal{P}_{\boldsymbol{\psi}}, \boldsymbol\Delta_{\boldsymbol\alpha},\emptyset)=1$.
\end{theorem}	      
\begin{remark}
	From the definition, the ratio on the left-hand side depends on metric $g$ only through its conformal class $[g]$.
\end{remark}

\begin{proof}
We give a new derivation from the local Ward identity.
	It suffices to prove the Ward identities for real $(\alpha_1,\alpha_2,\alpha_3)$; the general case follows from analytic continuation. 
    We first treat the case $\tilde\nu_j=\emptyset$.
	For $n\geq 1$, we choose a meromorphic vector field   that exhibits a $(n-1)$-th order pole only at point \(  x_1 \) and is holomorphic elsewhere:  $\mathrm{v}_{1}^{(n)}=-(z- x_1)^{-n+1}\partial_z$, with expansions:
	$$
    \psi_1^*(\mathrm{v}^{(n)}_{1})=-\sum_{k\geq -n} [\psi_1^*(\mathrm{v}^{(n)}_{
    1})]_k\cdot w^{k+1} \partial_w, \quad \text{ with } [\psi_1^*(\mathrm{v}^{(n)}_{1})]_{-n}=\frac{1}{\psi_1'(0)^n},
    $$
    by holomorphicity near $x_2$ and $x_3$,
	$$
    \psi_2^*(\mathrm{v}_{1}^{(n)})=-\sum_{k\geq -1} [\psi_2^*(\mathrm{v}_{1}^{(n)})]_k\cdot w^{k+1} \partial_w, \quad \text{ and } \psi_3^*(\mathrm{v}_{1}^{(n)})=-\sum_{k\geq -1} [\psi_3^*(\mathrm{v}_{1}^{(n)})]_k\cdot w^{k+1}\partial_w.
    $$
	Plugging $\mathrm{v}_1^{(n)} $ and $i\mathrm{v}_1^{(n)}$ into \Cref{lem:localward} and take the difference, it gives:
	\begin{align*}
		& -\frac{1}{{\psi_1^{'}(0)}^n} \mathcal{A}_{ \mathcal{P}_{\boldsymbol{\psi}} , g_{\mathrm{ad}}, \boldsymbol\psi} (\bL_{-n}\otimes \mathrm{Id}\otimes\mathrm{Id})\\
        = 
		&\sum_{k > -n} [\psi_1^*(\mathrm{v}^{(n)}_{1})]_k\cdot\mathcal{A}_{ \mathcal{P}_{\boldsymbol{\psi}} , g_{\mathrm{ad}}, \boldsymbol\psi}(\bL_{k}\otimes \mathrm{Id}\otimes\mathrm{Id})
		\\
		& + \sum_{k\geq -1}[\psi_2^*(\mathrm{v}_{1}^{(n)})]_k\cdot \mathcal{A}_{ \mathcal{P}_{\boldsymbol{\psi}} , g_{\mathrm{ad}}, \boldsymbol\psi} (\mathrm{Id}\otimes \mathbf{L}_{k}\otimes\mathrm{Id})+\sum_{k \geq -1} [\psi_3^*(\mathrm{v}_{1}^{(n)})]_k\cdot\mathcal{A}_{ \mathcal{P}_{\boldsymbol{\psi}} , g_{\mathrm{ad}}, \boldsymbol\psi} (\mathrm{Id}\otimes\mathrm{Id}\otimes \mathbf{L}_{k})\\
		& + \frac{\mathbf{i} c_{\rm L}}{24 \pi} \sum_{j=1}^3 \int_{\psi_j(\T)} \mathrm{v}_{1}^{(n)} S_{\psi_j^{-1}} \mathrm d z \cdot \mathcal{A}_{ \mathcal{P}_{\boldsymbol{\psi}} , g_{\mathrm{ad}}, \boldsymbol\psi}(\mathrm{Id}\otimes\mathrm{Id}\otimes \mathrm{Id}).
	\end{align*}
	From this recursion, we can reduce the computation of $\mathcal{A}_{ \mathcal{P}_{\boldsymbol{\psi}} , g_{\mathrm{ad}}, \boldsymbol\psi}(\otimes_{j=1}^3\Psi_{\alpha_j,\nu_j,\emptyset})$ to the computation of $\mathcal{A}_{ \mathcal{P}_{\boldsymbol{\psi}} , g_{\mathrm{ad}}, \boldsymbol\psi}(\Psi_{\alpha_1}\otimes(\otimes_{j=2}^3\Psi_{\alpha_j,\nu_j,\emptyset}))$. 
    
    Note that for admissible metric $g$ on $\D$ with $g_\D:=|\dd z|^2$ near the origin,
	\begin{equation}
		\mathcal{A}_{\psi_j(\D),(\psi_j)_*g,  x_j,\alpha_j,\psi_j}=e^{c_{\rm L} S_{\rm L}^0(\D,g_\D,g)}Z_{\D,g_\D} \Psi_{\alpha_j,\emptyset,\emptyset}.
	\end{equation}
	Then, we have $g_{\mathrm{ad}}\#(\psi_1)_*g$, which still denoted by $g_{\mathrm{ad}}$, as the admissible metric on $\mathbb{A}_{\boldsymbol\psi} = \mathbb{A}_{\psi_2,\psi_3}:=\hat\C\setminus(\psi_2(\D)\cup\psi_3(\D))$, with the boundary parameterization $\boldsymbol\psi=(\psi_2,\psi_3)$, then
	\begin{equation}
		\frac{1}{\sqrt{2}\pi}e^{c_{\rm L} S_{\rm L}^0(\D,g_\D,g)}Z_{\D,g_\D}
        \mathcal{A}_{ \mathcal{P}_{\boldsymbol{\psi}} , g_{\mathrm{ad}}, \boldsymbol\psi}(\Psi_{\alpha_1}\otimes(\otimes_{j=2}^3\Psi_{\alpha_j,\nu_j,\emptyset}))
        =
        \mathcal{A}_{ \mathbb{A}_{\boldsymbol\psi}, g_{\mathrm{ad}}, x_1,\alpha_1, \boldsymbol\psi}(\otimes_{j=2}^3\Psi_{\alpha_j,\nu_j,\emptyset}).
    \end{equation}
	Similarly, take  $\mathrm{v}_{2}^{(n)}=-(z- x_2)^{-n+1}\partial_z$, with expansions:
	$$
    \psi_2^*(\mathrm{v}_{2}^{(n)})=-\sum_{k\geq -n}[\psi_2^*(\mathrm{v}_{2}^{(n)})]_{k}\cdot z^{k+1}\partial_z, \quad \text{ with } [\psi_2^*(\mathrm{v}_{2}^{(n)})]_{-n}=\frac{1}{\psi_2'(0)^n},
    $$
    and 
    $$
    \psi_3^*(\mathrm{v}_{2}^{(n)})=-\sum_{k\geq -1} [\psi_3^*(\mathrm{v}_{2}^{(n)})]_{k}\cdot z^{k+1}\partial_z.
    $$
	Then, we have
	\begin{align*}
		-\frac{1}{\psi_2^{'}(0)^n} \mathcal{A}_{\mathbb{A}_{\boldsymbol\psi}, g_{\mathrm{ad}},  x_1,\alpha_1, \boldsymbol\psi}(\bL_{-n}\otimes \mathrm{Id})= 
		&\sum_{k > -n} [\psi_2^*(\mathrm{v}_{2}^{(n)})]_{k}\cdot\mathcal{A}_{\mathbb{A}_{\boldsymbol\psi}, g_{\mathrm{ad}},  x_1,\alpha_1, \boldsymbol\psi} (\bL_{k}\otimes \mathrm{Id}) \\
		&+ \sum_{k \geq -1} [\psi_3^*(\mathrm{v}_{2}^{(n)})]_{k}\cdot\mathcal{A}_{\mathbb{A}_{\boldsymbol\psi}, g_{\mathrm{ad}},  x_1,\alpha_1, \boldsymbol\psi} (\mathrm{Id}\otimes \mathbf{L}_{k})\\
		&- \mathbf{D^{1,23}}( x_1, \alpha_1 , \boldsymbol\psi, \mathrm{v}_{2}^{(n)}) \mathcal{A}_{\mathbb{A}_{\boldsymbol\psi}, g_{\mathrm{ad}},  x_1,\alpha_1, \boldsymbol\psi}( \mathrm{Id}\otimes  \mathrm{Id})
	\end{align*}
Here for $\mathbf{I}\cup \mathbf{J}=\{1,2,3\}$, $\mathbf{I}\cap \mathbf{J}=\emptyset$, $\mathbf{D^{I,J}}({\bf x},\bs\alpha,\mathrm{v})$ is the differential operator defined by
        \begin{align}
            \mathbf{D^{I,J}}({\bf x},\bs\alpha, \boldsymbol\psi,\mathrm{v})=\sum_{i\in I} \left( \Delta_{\alpha_i}(\partial_z\mathrm{v}(x_i)+\partial_z\sigma(x_i,\bar x_i)\cdot \mathrm{v}(x_i))+\mathrm{v}(x_i)\partial_{x_i} \right)-\frac{\mathrm{i}c_{\rm L}}{24\pi}\sum_{j\in J}\int_{\psi_j(\T)}\mathrm{v}S_{\psi^{-1}_j}\mathrm{d}z
        \end{align}

    \noindent
	From above, we reduce the computation to (also its derivative) 
    $$
    \frac{1}{\sqrt{2}\pi}e^{c_{\rm L} S_{\rm L}^0(\D,g_\D,g)}Z_{\D,g_\D}\mathcal{A}_{\mathbb{A}_{\boldsymbol\psi}, g_{\mathrm{ad}},  x_1,\alpha_1, \boldsymbol\psi}(\Psi_{\alpha_2}\otimes\Psi_{\alpha_3,\nu_3,\emptyset})=\mathcal{A}_{\mathbb{D}^{\mathsf{c}}_{\psi}, g_{\mathrm{ad}}, \bf x,\bs \alpha, \psi}(\Psi_{\alpha_3,\nu_3,\emptyset}),
    $$
	where $\mathbb{D}^{\mathsf{c}}_{\psi}=\hat\C\setminus\psi_3(\D)$, $\psi=\psi_3$. 
    
    Take $\mathrm{v}^{(n)}_{3}=-(z-x_3)^{-n+1}\partial_z$ with expansion 
    $$
    \psi_3^*(\mathrm{v}^{(n)}_{3})=-\sum_{k\geq -n} [\psi_3^*(\mathrm{v}^{(n)}_{3})]_{k}\cdot z^{k+1}\partial_z, \quad \text{ with }[\psi_3^*(\mathrm{v}^{(n)}_{3})]_{-n}=\frac{1}{\psi_3'(0)^n}.
    $$
	Then we have
	\begin{align}\label{e:disk ward identity analytic}
    \begin{split}
		-\frac{1}{\psi'_3(0)^n} \mathcal{A}_{\mathbb{D}^{\mathsf{c}}_{\psi}, g_{\mathrm{ad}}, \bf x,\bs  \alpha, \psi}(\bL_{-n})= 
		&\sum_{k > -n} [\psi_3^*(\mathrm{v}^{(n)}_{3})]_{k}\cdot\mathcal{A}_{\mathbb{D}^{\mathsf{c}}_{\psi}, g_{\mathrm{ad}}, \bf x,\bs  \alpha, \psi}(\bL_{k})
		\\-&
        \mathbf{D^{12,3}}({\bf x},\bs\alpha,\psi,\mathrm{v}^{(n)}_{3}) \cdot
        \mathcal{A}_{\mathbb{D}^{\mathsf{c}}_{\psi}, g_{\mathrm{ad}}, \bf x,\bs \alpha, \psi}(\mathrm{Id})
    \end{split}
	\end{align}
    In the end, by iterating the three recursion relations, we reduce every
holomorphic descendant coefficient to a finite linear combination of
compositions of the operators $\mathbf D^{I,J}$ applied to the primary
amplitude
\(
 \mathcal A_{\mathcal P_{\boldsymbol\psi},
 g_{\mathrm{ad}},\boldsymbol\psi}
 (\Psi_{\alpha_1},\Psi_{\alpha_2},\Psi_{\alpha_3}).
\)
At each step, the level on the boundary component under consideration
strictly decreases, while the Virasoro commutation relations produce
coefficients polynomial in
$\boldsymbol\Delta_{\boldsymbol\alpha}$. We therefore obtain
\eqref{e:holormorphic fact} when
$\widetilde\nu_j=\emptyset$ for all $j$. When anti-holomorphic
descendants are included, the commutative relation
\(
 [\mathbf L_n,\widetilde{\mathbf L}_m]=0
\)
shows that the holomorphic recursion leaves the anti-holomorphic
descendants unchanged. Applying the conjugate recursion to the
anti-holomorphic sector gives a second polynomial factor depending
only on $\boldsymbol{\widetilde\nu}$. This factor is
the complex conjugate of the corresponding holomorphic polynomial.
The claimed factorization follows, and the general case is obtained
by analytic continuation.
\end{proof}

As a consequence, we have the following Ward identities for 
\(\mathbb{D}^{\mathsf{c}}_{\psi}\) and \(\mathbb{A}_{\boldsymbol{\psi}}\),the matrix coefficients can also be written as $\omega(  \mathcal{P}_{\boldsymbol{\psi}}, \boldsymbol\Delta_{\boldsymbol\alpha} ,\bs\nu)$.
\begin{corollary}
	Retain the above notations, we have:
	\begin{enumerate}
		\item For $\bs\nu=(\emptyset,\nu_2,\nu_3)$, $\tilde{\bs\nu}=(\emptyset,\tilde\nu_2,\tilde\nu_3)$, \begin{equation}
			\frac{\mathcal{A}_{\mathbb{A}_{\boldsymbol\psi}, g_{\mathrm{ad}},  x_1,\alpha_1, \boldsymbol\psi}( \Psi_{\alpha_2,\nu_2,\tilde\nu_2}, \Psi_{\alpha_3,\nu_3,\tilde\nu_3})}{\mathcal{A}_{\mathbb{A}_{\boldsymbol\psi}, g_{\mathrm{ad}},  x_1,\alpha_1, \boldsymbol\psi}(\Psi_{\alpha_2}, \Psi_{\alpha_3})}=
        \omega(  \mathcal{P}_{\boldsymbol{\psi}} , \boldsymbol\Delta_{\boldsymbol\alpha} ,\bs\nu)\overline{\omega(  \mathcal{P}_{\boldsymbol{\psi}},\boldsymbol\Delta_{\boldsymbol\alpha},\bs{\tilde\nu})}.
		\end{equation}
		\item For $\bs\nu=(\emptyset,\emptyset,\nu_3)$, $\tilde{\bs\nu}=(\emptyset,\emptyset,\tilde\nu_3)$,
		\begin{equation}
			\frac{\mathcal{A}_{\mathbb{D}^{\mathsf{c}}_{\psi}, g_{\mathrm{ad}}, \bf x,\bf  \alpha, \psi}( \Psi_{\alpha_3,\nu_3,\tilde\nu_3})}{\mathcal{A}_{\mathbb{D}^{\mathsf{c}}_{\psi}, g_{\mathrm{ad}}, \bf x,\bf  \alpha, \psi}( \Psi_{\alpha_3,\emptyset,\emptyset})}=
        \omega(  \mathcal{P}_{\boldsymbol{\psi}} , \boldsymbol\Delta_{\boldsymbol\alpha} ,\bs\nu)\overline{\omega(  \mathcal{P}_{\boldsymbol{\psi}},\boldsymbol\Delta_{\boldsymbol\alpha},\bs{\tilde\nu})}.
		\end{equation}
	\end{enumerate}
\end{corollary}

\subsection{Relation to the Liouville vertex operator}\label{sec:liouville-vertex-operator}

We now compare the amplitude Ward identities with formal chiral vertex operators. Let $\alpha_j=Q+\mathrm iP_j$ with $P_j\in\mathbb R_+$ for $j=1,2,3$. By Proposition~\ref{liouville module}, the modules $\mathcal V_{\alpha_j}$ are irreducible highest-weight modules with central charge $c_{\rm L}$ and highest weights $\Delta_{\alpha_j}$. We denote their highest-weight vectors by $v_{\alpha_j}$ and use the partition convention $\mathcal T$.

\paragraph{Formal chiral vertex operators.}
Following \cite[Section~17.1]{Teschner_2001}, let $\widehat{\mathcal V}_{\alpha_3}:=\prod_{N\geq0}\mathcal V_{\alpha_3}[N]$ be the formal completion. There is a unique formal chiral vertex operator
\begin{equation}\label{eq:formal-chiral-vertex-map}
 \mathsf V_{\boldsymbol\alpha}(\,\cdot\,;z):
 \mathcal V_{\alpha_2}\longrightarrow
 \operatorname{Hom}\bigl(\mathcal V_{\alpha_1},\widehat{\mathcal V}_{\alpha_3}\bigr)\{z\},
 \qquad \boldsymbol\alpha=(\alpha_1,\alpha_2,\alpha_3),
\end{equation}
where $\{z\}$ denotes formal Laurent series with the overall power $z^{\Delta_{\alpha_3}-\Delta_{\alpha_2}-\Delta_{\alpha_1}}$. which is characterized by $\mathsf V_{\boldsymbol\alpha}(v_{\alpha_2};z)v_{\alpha_1}=z^{\Delta_{\alpha_3}-\Delta_{\alpha_2}-\Delta_{\alpha_1}}(v_{\alpha_3}+O(z))$ and, for $n\in\mathbb Z$ and $w\in\mathcal V_{\alpha_2}$,
\begin{equation}\label{eq:chiral-vertex-commutator}
 [\mathbf L_n,\mathsf V_{\boldsymbol\alpha}(w;z)]
 =\sum_{k=-1}^{\infty}\binom{n+1}{k+1}z^{n-k}
 \mathsf V_{\boldsymbol\alpha}(\mathbf L_kw;z).
\end{equation}
For $n<-1$, the binomial coefficient is understood in its generalized sense. For a fixed descendant $w$, the sum is finite because $\mathbf L_kw=0$ for all sufficiently large $k$. In particular, for $w=v_{\alpha_2}$, \eqref{eq:chiral-vertex-commutator} becomes $[\mathbf L_n,\mathsf V_{\boldsymbol\alpha}(v_{\alpha_2};z)]=z^n\bigl(z\partial_z+(n+1)\Delta_{\alpha_2}\bigr)\mathsf V_{\boldsymbol\alpha}(v_{\alpha_2};z)$.

\
We now explain the relation between the Liouville vertex operators and the Liouville amplitudes. 
Fix $q\in(0,1)$ and $z$ with $q<|z|<1$, and choose $r>0$ such that $\overline{B(z,r)}\subset\mathbb A_q$. Set $\Sigma_{q,r,z}:=\mathbb A_q\setminus\overline{B(z,r)}$ and parametrize its inner, outer, and circular-hole boundaries by $\zeta_1(w)=qw$, $\zeta_2(w)=w$, and $\zeta_3(w)=z+rw$, respectively. Thus $\zeta_1$ and $\zeta_3$ are incoming, whereas $\zeta_2$ is outgoing. We write $\boldsymbol\zeta=(\zeta_1,\zeta_2,\zeta_3)$ and take $g$ to be an admissible metric for these parametrizations.

Let $o:\mathbb T\to\mathbb T$ be given by $o(e^{\mathrm i\theta})=e^{-\mathrm i\theta}$, let $\mathbf OF(\varphi):=F(\varphi\circ o)$, and let $\mathbf C$ be complex conjugation on $\mathcal H$. Following~\cite[Section~6.6]{segalaxiom}, let
$\mathbf J:=\mathbf O\mathbf C$ denote the anti-linear involution
associated with orientation reversal and complex conjugation. By
\cite[Proposition~6.11]{segalaxiom},
\(
 \mathbf J\Psi_{Q+\mathrm iP,\nu,\widetilde\nu}
 =
 \Psi_{Q-\mathrm iP,\nu,\widetilde\nu}.
\)
For finite descendant vectors $u_1,u_2,u_3$, as we view the amplitude as an operator
\[
 \mathcal A_\Sigma:
 \mathcal H_{\rm in}\otimes\mathcal H_{\rm in}
 \longrightarrow
 \mathcal H_{\rm out}.
\] we have
\begin{equation}\label{eq:outgoing-incoming-pairing}
 \left\langle
 \mathcal A_{\Sigma_{q,r,z},g,\boldsymbol\zeta}
 (u_1\otimes u_2),u_3
 \right\rangle_{\mathcal H}
 =
 \mathcal A_{\Sigma_{q,r,z},g,\boldsymbol\zeta^{\rm in}}
 \bigl(u_1\otimes u_2\otimes\mathbf J u_3\bigr),
\end{equation}
where
$\boldsymbol\zeta^{\rm in}:=(\zeta_1,\zeta_3,\zeta_2\circ o)$
lists the three boundary parametrizations in incoming order.
Together with the fact that $\mathbf{O}\mathbf{C}\mathbf{L}_n=\mathbf{O}\mathbf{C}\mathbf{L}^*_{-n}$, we have
\begin{equation}\label{eq:amplitude-vertex-intertwining}
 \mathbf L_n\mathcal A_{\Sigma_{q,r,z},g,\boldsymbol\zeta}
 -q^n\mathcal A_{\Sigma_{q,r,z},g,\boldsymbol\zeta}
 \circ(\mathbf L_n\otimes\mathrm{Id})
 =\sum_{k=-1}^{\infty}\binom{n+1}{k+1}z^{n-k}r^k
 \mathcal A_{\Sigma_{q,r,z},g,\boldsymbol\zeta}
 \circ(\mathrm{Id}\otimes\mathbf L_k).
\end{equation}

\begin{proposition}[Amplitude realization of the chiral vertex operator]
\label{prop:amplitude-chiral-vertex}
For \(\nu_1,\nu_2,\nu_3\in\mathcal{T}\), fix a branch of
\(z^{\Delta_{\alpha_3}-\Delta_{\alpha_2}-\Delta_{\alpha_1}}\), one has
\begin{align}
&\left\langle
\mathsf{V}_{\boldsymbol{\alpha}}
\bigl(\mathbf{L}_{-\nu_2}v_{\alpha_2};z\bigr)
\mathbf{L}_{-\nu_1}v_{\alpha_1},
\mathbf{L}_{-\nu_3}v_{\alpha_3}
\right\rangle_{\alpha_3}^{\mathrm{Sh}}
=
z^{\Delta_{\alpha_3}-\Delta_{\alpha_2}-\Delta_{\alpha_1}}
\frac{
\left\langle
\mathcal{A}_{\Sigma_{q,r,z},g,\boldsymbol{\zeta}}
\bigl(
\Psi_{\alpha_1,\nu_1,\emptyset}
\otimes
\Psi_{\alpha_2,\nu_2,\emptyset}
\bigr),
\Psi_{\alpha_3,\nu_3,\emptyset}
\right\rangle
}{
q^{|\nu_1|}
r^{|\nu_2|}
\left\langle
\mathcal{A}_{\Sigma_{q,r,z},g,\boldsymbol{\zeta}}
\bigl(\Psi_{\alpha_1}\otimes\Psi_{\alpha_2}\bigr),
\Psi_{\alpha_3}
\right\rangle
}.
\label{eq:amplitude-chiral-vertex}
\end{align}
Equivalently, the two sides are
\(
C_{\boldsymbol{\alpha}}(\nu_1,\nu_2;\nu_3)
z^{
\Delta_{\alpha_3}
-\Delta_{\alpha_2}
-\Delta_{\alpha_1}
+|\nu_3|
-|\nu_2|
-|\nu_1|
},
\)
where
\(C_{\boldsymbol{\alpha}}(\nu_1,\nu_2;\nu_3)\)
is a polynomial in
\((\Delta_{\alpha_1},\Delta_{\alpha_2},\Delta_{\alpha_3})\)
of total degree at most
\(\ell(\nu_1)+\ell(\nu_2)+\ell(\nu_3)\).
\end{proposition}

\begin{proof}
Since the factors $q^n$ and $r^k$ in
\eqref{eq:amplitude-vertex-intertwining} are exactly cancelled by the
corresponding changes of level, the resulting right-hand side of \eqref{eq:amplitude-chiral-vertex} satisfies
the commutation relation~\eqref{eq:chiral-vertex-commutator}. Its value
on the three highest-weight vectors is
\(
 z^{\Delta_{\alpha_3}-\Delta_{\alpha_2}-\Delta_{\alpha_1}}.
\)
The uniqueness of formal chiral vertex-operator matrix
coefficients therefore gives~\eqref{eq:amplitude-chiral-vertex}.
\end{proof}
\begin{corollary}[M\"obius parametrizations of a pair of pants]
Let
\[
  \psi_j:U_j\longrightarrow \widehat{\mathbb C},
  \qquad j=1,2,3,
\]
be univalent holomorphic maps on neighborhoods $U_j$ of
$\overline{\mathbb D}$. Assume that the compact
sets $\psi_j(\overline{\mathbb D})$ are pairwise disjoint, and set
\[
  \mathcal P_{\boldsymbol\psi}
  :=\widehat{\mathbb C}\setminus
  \bigcup_{j=1}^3\psi_j(\mathbb D),
  \qquad
  \boldsymbol\psi:=(\psi_1,\psi_2,\psi_3),
\]
with all three boundary parametrizations incoming. Put
\[
  h_j:=\Delta_{\alpha_j},
  \qquad
  \check\alpha_3:=2Q-\alpha_3,
  \qquad
  N_j:=|\nu_j|.
\]
Define
\[
\begin{aligned}
  \Omega_{\boldsymbol\psi}(\nu_1,\nu_2,\nu_3)
  :=\frac{
  \mathcal A_{\mathcal P_{\boldsymbol\psi},g,\boldsymbol\psi}
  \bigl(
    \Psi_{\alpha_1,\nu_1,\emptyset},
    \Psi_{\alpha_2,\nu_2,\emptyset},
    \Psi_{\check\alpha_3,\nu_3,\emptyset}
  \bigr)
  }{
  \mathcal A_{\mathcal P_{\boldsymbol\psi},g,\boldsymbol\psi}
  \bigl(
    \Psi_{\alpha_1},
    \Psi_{\alpha_2},
    \Psi_{\check\alpha_3}
  \bigr)
  }.
\end{aligned}
\]
\begin{enumerate}[(i)]
\item Suppose that
\[
  \psi_1(w)=qw,
  \qquad
  \psi_2(w)=1+rw,
  \qquad
  \psi_3(w)=\frac{R}{w},
\]
where $q,r,R>0$, $q+r<1$, and $1+r<R$. Then
\begin{equation}
  C_{\boldsymbol\alpha}(\nu_1,\nu_2;\nu_3)
  =
  \frac{
  \mathcal A_{\mathcal P_{\boldsymbol\psi},g,\boldsymbol\psi}
  \bigl(
    \Psi_{\alpha_1,\nu_1,\emptyset},
    \Psi_{\alpha_2,\nu_2,\emptyset},
    \Psi_{\check\alpha_3,\nu_3,\emptyset}
  \bigr)
  }{
  q^{N_1}r^{N_2}R^{-N_3}
  \mathcal A_{\mathcal P_{\boldsymbol\psi},g,\boldsymbol\psi}
  \bigl(
    \Psi_{\alpha_1},
    \Psi_{\alpha_2},
    \Psi_{\check\alpha_3}
  \bigr)
  }.
\label{eq:standard-mobius-pants-coefficient}
\end{equation}

\item Suppose that each $\psi_j$ is a M\"obius embedding. Let
$x_j:=\psi_j(0)$ and let $M\in\operatorname{PSL}(2,\mathbb C)$ be the
unique M\"obius transformation satisfying
\[
  M(x_1)=0,
  \qquad
  M(x_2)=1,
  \qquad
  M(x_3)=\infty.
\]
Set $\phi_j:=M\circ\psi_j$. There are unique parameters
\[
  \rho_1,\rho_2,\rho_3\in\mathbb C^\times,
  \qquad
  \kappa_1,\kappa_2,\kappa_3\in\mathbb C,
\]
such that
\begin{equation}
  \phi_1(w)=\frac{\rho_1w}{1-\kappa_1w},\qquad
  \phi_2(w)=1+\frac{\rho_2w}{1-\kappa_2w},\qquad
  \phi_3(w)=\frac{\rho_3}{w}+\sigma_3
             =\frac{\rho_3}{w}(1-\kappa_3w),
  \qquad \kappa_3=-\frac{\sigma_3}{\rho_3}.
\label{eq:mobius-germ-normal-forms}
\end{equation}
Equivalently,
\[
  \rho_j=\phi_j'(0),
  \qquad
  \kappa_j=\frac{\phi_j''(0)}{2\phi_j'(0)}
  \quad (j=1,2),
  \qquad
  \rho_3=\lim_{w\to0}w\phi_3(w).
\]
For partitions $\nu,\mu$, define, in the universal Verma module, the
universal lower-triangular coefficients $\theta_{\nu,\mu}(h)$ by
\begin{equation}
  e^{\kappa L_1}L_{-\nu}v_h
  =
  \sum_{\substack{\mu\in\mathcal T\\ |\mu|\le |\nu|}}
  \kappa^{|\nu|-|\mu|}
  \theta_{\nu,\mu}(h)L_{-\mu}v_h.
\label{eq:theta-coordinate-change}
\end{equation}
Thus the sum is finite and, whenever $|\mu|=|\nu|$,
$\theta_{\nu,\mu}(h)=\delta_{\nu,\mu}$. Then
\begin{equation}
  \Omega_{\boldsymbol\psi}(\nu_1,\nu_2,\nu_3)
  ={}
  \sum_{\substack{
       \mu_j\in\mathcal T,\ |\mu_j|\le N_j\\
       j=1,2,3}}
  \rho_1^{|\mu_1|}
  \rho_2^{|\mu_2|}
  \rho_3^{-|\mu_3|}\times
  \prod_{j=1}^3
  \Bigl[
    \kappa_j^{N_j-|\mu_j|}
    \theta_{\nu_j,\mu_j}(h_j)
  \Bigr]
  C_{\boldsymbol\alpha}(\mu_1,\mu_2;\mu_3).
\label{eq:general-mobius-pants-coefficient}
\end{equation}
\end{enumerate}
\end{corollary}

\begin{proof}
Formula ~\eqref{eq:standard-mobius-pants-coefficient} is a direct consequence of previous proposition.
We now prove~\eqref{eq:general-mobius-pants-coefficient}. By
diffeomorphism covariance, it suffices to work with
$\boldsymbol\phi=(\phi_1,\phi_2,\phi_3)$.
Choose any $q_0,r_0,R_0>0$ with $q_0+r_0<1$ and $1+r_0<R_0$, and set
\[
  \eta_1(w)=q_0w,
  \qquad
  \eta_2(w)=1+r_0w,
  \qquad
  \eta_3(w)=\frac{R_0}{w}.
\]
Then
\[
  \phi_j=\eta_j\circ g_j,
  \qquad
  g_j(w)=\frac{\lambda_jw}{1-\kappa_jw}
        =d_{\lambda_j}\circ f_{\kappa_j}(w),
\]
where
\[
  f_\kappa(w):=\frac{w}{1-\kappa w},
  \qquad
  d_\lambda(w):=\lambda w,
\]
and
\[
  \lambda_1=\frac{\rho_1}{q_0},
  \qquad
  \lambda_2=\frac{\rho_2}{r_0},
  \qquad
  \lambda_3=\frac{R_0}{\rho_3}.
\]
The family $f_\kappa$ is the flow of $w^2\partial_w$. In the mode
convention
\[
  V(w)\partial_w=-\sum_m[V]_mw^{m+1}\partial_w,
\]
this vector field has $[V]_1=-1$. By subdividing the path
$t\mapsto f_{t\kappa}$ into sufficiently small pieces, and then using
holomorphic continuation in $\kappa$, the boundary-variation formula,
whose Virasoro term carries an overall minus sign, gives
\[
  \mathcal A_{\eta\circ f_\kappa}(u)
  =\mathcal A_\eta(e^{\kappa L_1}u).
\]
There is no scalar term here because all relevant germs are M\"obius,
their Schwarzians vanish, and the vector field has no zero mode. For a
dilation, the state-dependent part of the normalized coordinate change
is $\lambda^{L_0-h}$; the remaining scalar factor is independent of the
inserted descendant and cancels in the normalized ratio. Consequently,
for each boundary component the normalized coordinate-change operator
is
\[
  T_j=\lambda_j^{L_0-h_j}e^{\kappa_jL_1}.
\]
Using~\eqref{eq:theta-coordinate-change},
\[
  T_jL_{-\nu_j}v_{h_j}
  =\sum_{|\mu_j|\le N_j}
    \lambda_j^{|\mu_j|}
    \kappa_j^{N_j-|\mu_j|}
    \theta_{\nu_j,\mu_j}(h_j)L_{-\mu_j}v_{h_j}.
\]
Insert these three finite expansions into the amplitude with reference
charts $\boldsymbol\eta$, and apply
\eqref{eq:standard-mobius-pants-coefficient}. The factors from the
reference charts combine with the dilations as
\[
  q_0^{|\mu_1|}\lambda_1^{|\mu_1|}
  =\rho_1^{|\mu_1|},
  \qquad
  r_0^{|\mu_2|}\lambda_2^{|\mu_2|}
  =\rho_2^{|\mu_2|},
\qquad
  R_0^{-|\mu_3|}\lambda_3^{|\mu_3|}
  =\rho_3^{-|\mu_3|}.
\]
This proves~\eqref{eq:general-mobius-pants-coefficient}.
\end{proof}
\section{Applications}

In this section we give two applications of the variational formula for
amplitudes. First, we prove smoothness of Liouville correlation functions
with respect to the positions of bulk insertions. Second, on the Riemann sphere, we
derive BPZ equations for correlation functions with a bulk degenerate
insertion; the level-two case is written out explicitly, and the general
\((r,s)\) case follows by the Ward-identity reduction.

\subsection{Smoothness of Liouville correlations}

On the Riemann sphere, smoothness of Liouville correlation functions with
respect to the bulk marked points was proved by Oikarinen
in~\cite{Oikarinen_2019}. For compact Riemann surfaces, the corresponding
statement is established in
\cite[Proposition~5.2]{oikarinen2022stressenergyliouvilleconformalfield}.
Both arguments rely on fusion estimates and Gaussian multiplicative chaos
bounds. We provide an alternative geometric proof based on the local Ward
identity. The main point is that derivatives with respect to the insertion
variables are directly identified with Virasoro insertions and explicit
anomaly terms. Higher differentiability is therefore reduced to the continuity of
amplitudes with finitely many higher Virasoro insertions. The only additional
analytic input is the following weighted extension of the boundary-variation
formula, which was already established (implicityly) in~\cite{segalaxiom,semigroup}. We recall the following concrete version.

\begin{lemma}[Weighted extension of the local Ward identity]
\label{lem:weighted-local-ward}
Let $(\Sigma,g_{\rm ad},\mathbf x,\boldsymbol\alpha,\zeta)$ be an
admissible surface of type $(\mathfrak g,n,0,1)$, let
\(
 K\subset\operatorname{Conf}_n(\Sigma^\circ)
\)
a compact subset,
\(
 s_\Sigma:=\sum_{i=1}^n\alpha_i-Q\chi(\Sigma),
\)
and assume that $s_\Sigma>Q$.
Here \(
 \operatorname{Conf}_n(X)
 :=
 \left\{
 (x_1,\ldots,x_n)\in X^n:
 x_i\neq x_j\text{ for }i\neq j
 \right\}.
\)
Fix
\(
 Q<\beta<s_\Sigma
\)
and set
\[
 \mathcal V_0^{\rm fin}
 :=
 \operatorname{Span}\left\{
 \Psi_{0,\nu,\widetilde\nu}:
 \nu,\widetilde\nu\in\mathcal T
 \right\} \subset e^{-\beta c_-}\mathcal D(\mathcal Q).
\]
Then, for every $U\in\mathcal V_0^{\rm fin}$, the map
\[
 \mathbf x\longmapsto
 \mathcal A_{\Sigma,g_{\rm ad},
 \mathbf x,\boldsymbol\alpha,\zeta}(U)
\]
is continuous on $K$.

Let $\mathbf x\mapsto v_{\mathbf x}$ be a smooth family of
meromorphic vector fields on the filling surface $\widehat\Sigma$, whose
poles lie in a fixed filling disk of bounded order,
and set
\(
 \mathbf x^{\,t}
 :=
 f_{-t}^{v_{\mathbf x}}(\mathbf x).
\)
Then, for every $U\in\mathcal V_0^{\rm fin}$,
the map
\[
 t\longmapsto
 \mathcal A_{\Sigma,g_{\rm ad},
 \mathbf x^{\,t},\boldsymbol\alpha,\zeta}(U)
\]
is differentiable at $t=0$, and
\begin{equation}\label{eq:weighted-evaluation-derivative}
 \left.
 \partial_t
 \mathcal A_{\Sigma,g_{\rm ad},
 \mathbf x^{\,t},\boldsymbol\alpha,\zeta}(U)
 \right|_{t=0}
 =
 \left(
 \left.
 \partial_t
 \mathcal A_{\Sigma,g_{\rm ad},
 \mathbf x^{\,t},\boldsymbol\alpha,\zeta}
 \right|_{t=0}
 \right)(U).
\end{equation}
The analogous statement holds for the anti-holomorphic variation.
\end{lemma}

\begin{proof}
The first statement is directly following from
\cite[Theorem~4.4 and Propositions~6.4,~11.13]{segalaxiom},
together with the sewing property.

For the second statement, choose
\(
Q<\beta<\beta'<s_\Sigma.
\)
By \eqref{eq:core-annulus-factorization} and diffeomorphism invariance, the family of amplitudes $\mathcal A_{\Sigma,g_{\rm ad},
 \mathbf x^{\,t},\boldsymbol\alpha,\zeta}=\mathcal A_{f_{t}^{v_{\mathbf x}}(\Sigma)}$ is obtained by
gluing the amplitude $\mathcal A_{\Sigma^0}$ associated with $\Sigma^0$ with
the annulus amplitude $\mathcal A_{\mathbb A_{rf_j}}$. By \cite[Section~5]{segalaxiom},
\(
\mathcal A_{\Sigma^0}
\in e^{\beta'c_-}\mathcal H
\subset e^{\beta c_-}\mathcal H.
\)
The Gaussian zero-mode estimates give the corresponding control of the
exponential weights, while \cite[Lemma~4.10]{semigroup} gives the
$\mathcal D(\mathcal Q)$-regularity through the factorization
$T_f=T_f^\eta e^{-\eta H}$.
Thus, we have, for $Q<\beta<\beta'<s_\Sigma$,
\[
\bigl|
(\mathcal A_{\Sigma^0}\circ
\mathcal A_{\mathbb A_{rf_j}})(F)
\bigr|
\leq
C_{\beta,\beta'}
\bigl\|
e^{-\beta'c_-}\mathcal A_{\Sigma^0}
\bigr\|_{\mathcal H}
\,
\|F\|_{e^{-\beta c_-}\mathcal D(\mathcal Q)},
\qquad
F\in e^{-\beta c_-}\mathcal D(\mathcal Q).
\]
Hence $\mathcal A_{f_{t}^{v_{\mathbf x}}(\Sigma)}$ belongs to
\(
\bigl(e^{-\beta c_-}\mathcal D(\mathcal Q)\bigr)'.
\)
The differentiability of this family follows from \cite[Theorem~4.8]{semigroup} and \eqref{eq:core-annulus-factorization}. Together with same argument, we have the differential $\partial_t\mathcal A_{f_{t}^{v_{\mathbf x}}(\Sigma)}$ also belongs to
\(
\bigl(e^{-\beta c_-}\mathcal D(\mathcal Q)\bigr)'.
\)
This proves the whole statement.
\end{proof}

\begin{theorem}\label{thm:smooth corrs}
Let $(\widehat\Sigma,\widehat g)$ be a closed Riemann surface of genus
$\mathfrak g$, let
\(
 \mathbf x=(x_1,\ldots,x_n)
 \in\operatorname{Conf}_n(\widehat\Sigma),
\)
and let
$\boldsymbol\alpha=(\alpha_1,\ldots,\alpha_n)\in\mathbb R^n$ satisfy
the Seiberg bounds
\[
 \sum_{j=1}^n\alpha_j>Q\chi(\widehat\Sigma),
 \qquad
 \alpha_j<Q\quad(1\leq j\leq n).
\]
Then the Liouville correlation function
\[
 \mathcal G(\mathbf x)
 :=
 \left\langle
 \prod_{i=1}^nV_{\alpha_i}(x_i)
 \right\rangle_{\widehat\Sigma,\widehat g}
\]
is smooth on
\(
 \operatorname{Conf}_n(\widehat\Sigma).
\)
\end{theorem}

\begin{proof}
Fix
$\mathbf x\in K\subset \operatorname{Conf}_n(\widehat\Sigma)$, with $K$ compact. After choosing
pairwise disjoint coordinate neighborhoods of the points $x_i$, choose
$p\in\widehat\Sigma$ and its coordinate disk $(D,w)$ outside these neighborhoods. Since $V_0(p)=1$, adding a zero-weight insertion does
not change the correlation:
\[
 \left\langle
 V_0(p)\prod_{i=1}^nV_{\alpha_i}(x_i)
 \right\rangle_{\widehat\Sigma,\widehat g}
 =\mathcal G(\mathbf x).
\]
Also the Seiberg bounds are preserved, since the total charge is unchanged and
$0<Q$.

Set
\(
 \Sigma:=\widehat\Sigma\setminus D^\circ
\)
and equip its boundary with an incoming parametrization $\zeta$. Choose an admissible metric $g_{\rm ad}$ on $\Sigma$. By Weyl covariance, there is a
constant $C_{\rm fill}$ independent of $\mathbf x\in K$, such that
\(
 \mathcal G(\mathbf x)
 =C_{\rm fill}\,
 \mathcal A_{\Sigma,g_{\rm ad},
 \mathbf x,\boldsymbol\alpha,\zeta}(\Psi_0).
\)
We absorb $C_{\rm fill}$ into the amplitude for the remainder of the
proof.

Since $\chi(\Sigma)=\chi(\widehat\Sigma)-1$, one has
\(
 s_\Sigma
 :=\sum_{i=1}^n\alpha_i-Q\chi(\Sigma)
 =Q+\sum_{i=1}^n\alpha_i-Q\chi(\widehat\Sigma)>Q.
\)
We fix
\(
 Q<\beta<s_\Sigma,
\)
and Lemma~\ref{lem:weighted-local-ward} applies to $\Psi_0$ and all of
its finite descendants. We then choose a smooth family
$\mathbf x\mapsto v_{\mathbf x}$ of meromorphic vector fields, such that $v_{\mathbf x}$
only has pole at $p$ of order at most $N$,
\[
 v_{\mathbf x}^1(x_1)=1,
 \qquad
 v_{\mathbf x}^i(x_i)=0
 \quad(2\leq i\leq n),
\]
where $v_{\mathbf x}=v_{\mathbf x}^i(w_i)\partial_{w_i}$ near $x_i$. Riemann--Roch theorem applied to
\(T\widehat{\Sigma}\otimes\mathcal O(Np-\sum_{i=2}^n x_i)\) gives the existence of
 vector fields. 
Write
\[
 \zeta^*v_{\mathbf x}
 =-\sum_{m > -N}
 [v_{\partial,\mathbf x}]_m
 w^{m+1}\partial_w,
 \qquad
 \mathbf L[v_{\partial,\mathbf x}]
 :=\sum_{m > -N}
 [v_{\partial,\mathbf x}]_m\mathbf L_m.
\]

By Lemma~\ref{lem:weighted-local-ward}, differentiation commutes with
evaluation at $\Psi_0$. Applying Lemma~\ref{lem:localward} to the
variations generated by $v_{\mathbf x}$ and
$\mathrm i v_{\mathbf x}$ and collecting them gives
\begin{align}
 \partial_{w_1}\mathcal G(\mathbf x)
 ={}&
 \mathcal A_{\Sigma,g_{\rm ad},
 \mathbf x,\boldsymbol\alpha,\zeta}
 \bigl(\mathbf L[v_{\partial,\mathbf x}]\Psi_0\bigr)
 -
 \sum_{i=1}^n\Delta_{\alpha_i}
 \left(
 \partial_{w_i}v_{\mathbf x}^i
 +(\partial_{w_i}\sigma_i)v_{\mathbf x}^i
 \right)(x_i)\mathcal G(\mathbf x)
 +c_{\rm L}\,
 \mathfrak a_{g_{\rm ad}}(v_{\mathbf x})
 \mathcal G(\mathbf x),
 \label{eq:first-smoothness-derivative}
\end{align}
where
$g_{\rm ad}=e^{\sigma_i(w_i,\bar w_i)}|dw_i|^2$ near $x_i$ and
\(
 \mathfrak a_{g_{\rm ad}}(v)
 :=
 \frac12
 \left.
 \left(
 \partial_tS_{\rm L}^0(\Sigma_t^v)
 -\mathrm i\partial_tS_{\rm L}^0(\Sigma_t^{\mathrm i v})
 \right)
 \right|_{t=0}
 -\frac1{24}[v_\partial]_0.
\)
Each term on the right-hand side of
\eqref{eq:first-smoothness-derivative} is continuous on $K$. Indeed,
since $\Psi_0$ is a highest-weight state and
$v_{\mathbf x}$ has a pole of uniformly bounded order at $p$,
\[
 \mathbf L[v_{\partial,\mathbf x}]\Psi_0
 =\sum_{m=1}^{N}
 [v_{\partial,\mathbf x}]_{-m}
 \mathbf L_{-m}\Psi_0.
\]
The coefficients in this
finite sum are smooth in $\mathbf x$, and the corresponding descendant
amplitudes are continuous by
Lemma~\ref{lem:weighted-local-ward}. The conformal-weight term is smooth.
Finally, 
\(
 \mathfrak a_{g_{\rm ad}}(v_{\mathbf x})
\)
also depends smoothly on $\mathbf x$.
Thus $\partial_{w_1}\mathcal G$ is continuous. The anti-holomorphic Ward
identity gives the continuity of $\partial_{\bar w_1}\mathcal G$.
Repeating the construction by choosing different vector field for each marked point shows that
$\mathcal G\in C^1(K)$.

It remains to iterate the argument. Choose an open neighborhood
$\mathcal O$ such that
\(
 K\subset\mathcal O
 \subset \operatorname{Conf}_n(\widehat\Sigma),
\)
and let $\mathcal E(\mathcal O)$ be the
$C^\infty(\mathcal O)$-module generated by the following:
\[
 \mathcal E(\mathcal O)
 :=
 \left\{
 \sum_{a=1}^M
 f_a(\mathbf x)\,
 \mathcal A_{\Sigma,g_{\rm ad},
 \mathbf x,\boldsymbol\alpha,\zeta}(U_a):
 f_a\in C^\infty(\mathcal O),\
 U_a\in\mathcal V_0^{\rm fin}
 \right\}.
\]
For every $U\in\mathcal V_0^{\rm fin}$,
\eqref{eq:first-smoothness-derivative} together with its anti-holomorphic version express each derivative
\[
 \partial_{w_i}
 \mathcal A_{\Sigma,g_{\rm ad},
 \mathbf x,\boldsymbol\alpha,\zeta}(U),
 \qquad
 \partial_{\bar w_i}
 \mathcal A_{\Sigma,g_{\rm ad},
 \mathbf x,\boldsymbol\alpha,\zeta}(U),
\]
as a finite linear combination of descendant amplitudes with
coefficients in $C^\infty(\mathcal O)$. Hence, we have
\[
 \partial_{w_i}\mathcal E(\mathcal O)
 \subset\mathcal E(\mathcal O),
 \qquad
 \partial_{\bar w_i}\mathcal E(\mathcal O)
 \subset\mathcal E(\mathcal O).
\]
Lemma~\ref{lem:weighted-local-ward} also gives
$\mathcal E(\mathcal O)\subset C^0(\mathcal O)$. Note that
\(
 \mathcal G(\mathbf x)
 =
 \mathcal A_{\Sigma,g_{\rm ad},
 \mathbf x,\boldsymbol\alpha,\zeta}(\Psi_0)
 \in\mathcal E(\mathcal O),
\)
 and every mixed
derivative of $\mathcal G$ belongs to $\mathcal E(\mathcal O)$, which is continuous. This gives $\mathcal G\in C^\infty(K)$. Since the
configuration and the neighborhood $K$ were arbitrary,
$\mathcal G$ is smooth on
$\operatorname{Conf}_n(\widehat\Sigma)$. This completes the proof.
\end{proof}

\begin{remark}[Boundary LCFT]
The same strategy applies to bulk marked points in boundary Liouville
correlation functions, provided the corresponding boundary amplitude
satisfies the weighted differentiability and continuity estimates used in
Lemma~\ref{lem:weighted-local-ward}.
\end{remark}

\subsection{Genus-zero BPZ equations}
The Belavin--Polyakov--Zamolodchikov equation follows from the null-vector
relation in a degenerate Virasoro representation.  We first spell out the
level-two case and then extend the Ward-identity reduction to
the general \((r,s)\) BPZ equations on the Riemann sphere.

Let $x,x_1,\ldots,x_N$ be pairwise distinct points in an affine coordinate
on $\widehat{\mathbb C}$, and let
$\boldsymbol\beta=(\beta_1,\ldots,\beta_N)$ be the weights attached to
$\mathbf x=(x_1,\ldots,x_N)$. After an affine rescaling, choose the round
boundary parametrization $\zeta_x(w)=x+w$ so that its filling disk is
disjoint from $\mathbf x$. We work either in the Seiberg region.
Choose $g_{\mathrm{ad}}=e^{\sigma(z,\bar z)}|dz|^2$ so that its
filling agrees with $\widehat g$ away from the cap, and set
\[
 \mathcal F_\alpha(x;\mathbf x)
 :=\mathcal A_{\mathbb D_x^{\mathsf c},g_{\mathrm{ad}},
 \mathbf x,\boldsymbol\beta,\zeta_x}(\Psi_\alpha).
\]
For $m\geq1$, let
\( \mathbf D_m^{(x)}
 :=\mathbf D_{m,\mathbf x,\boldsymbol\beta}(x),
\)
where $\mathbf D_{m,\mathbf x,\boldsymbol\beta}(x)$ is defined in
\eqref{eq:round-differential-operator}. In particular,
one has
\begin{align}
 \mathbf D_1^{(x)}
 &=-\sum_{i=1}^N \partial_{x_i}
 +\Delta_{\beta_i}\partial_{x_i}\sigma(x_i,\bar x_i),
 \nonumber\\
 \mathbf D_2^{(x)}
 &=\sum_{i=1}^N\left(
 \frac{1}{x-x_i} \left(\partial_{x_i}
 +\Delta_{\beta_i}\partial_{x_i}\sigma(x_i,\bar x_i) \right)
 +\frac{\Delta_{\beta_i}}{(x-x_i)^2}
 \right).
 \label{eq:level-two-ward-operators}
\end{align}

\begin{proposition}[Level-two BPZ equation on the sphere]
\label{prop:level-two-bpz}
Let $\alpha\in\{\alpha_{1,2},\alpha_{2,1}\}$ and assume that the
corresponding generalized amplitudes are defined at the stated weights.
Then
\begin{equation}
 \left(
 (\mathbf D_1^{(x)})^2
 +\alpha^2\mathbf D_2^{(x)}
 \right)\mathcal F_\alpha(x;\mathbf x)=0.
 \label{eq:level-two-bpz-amplitude}
\end{equation}
The same equation holds for the filled sphere correlation
\[
 \mathcal G_\alpha(x;\mathbf x)
 :=\left\langle
 V_\alpha(x)\prod_{i=1}^NV_{\beta_i}(x_i)
 \right\rangle_{\widehat{\mathbb C},\widehat g}.
\]
\end{proposition}

\begin{proof}
Applying \Cref{prop:matrixcoe}\,(i) gives
\[
 \mathcal A(\mathbf L_{-m}u)
 =\mathbf D_m^{(x)}\mathcal A(u),
 \qquad m\geq1,
\]
for every finite holomorphic descendant $u$. Hence
\[
 \mathcal A(\mathbf L_{-2}\Psi_\alpha)
 =\mathbf D_2^{(x)}\mathcal F_\alpha,
 \qquad
 \mathcal A(\mathbf L_{-1}^2\Psi_\alpha)
 =(\mathbf D_1^{(x)})^2\mathcal F_\alpha.
\]
Applying the amplitude to the null-vector relation
\eqref{eq:level-two-null-vector} proves
\eqref{eq:level-two-bpz-amplitude}. By the Weyl anomaly,
$\mathcal G_\alpha=C_{\mathrm{fill}}\mathcal F_\alpha$, where
$C_{\mathrm{fill}}$ depends only on the fixed filling data and is
independent of the marked points $\mathbf x$. The same operator
therefore annihilates $\mathcal G_\alpha$.
\end{proof}

By Weyl covariance, we pass to the flat-coordinate representative of
the correlation in an affine chart equipped with the metric
$|dz|^2$, and denote it by
$\mathcal G_\alpha^{\mathrm{flat}}$. In this representation, the
terms involving $\sigma$ in
\eqref{eq:level-two-ward-operators} disappear, and
$\mathcal G_\alpha^{\mathrm{flat}}$ is M\"obius covariant. In
particular, for the translation $f_t(z)=z+t$, whose derivative is
identically equal to one, M\"obius covariance gives
\[
 \mathcal G_\alpha^{\mathrm{flat}}
 (x+t;x_1+t,\ldots,x_N+t)
 =
 \mathcal G_\alpha^{\mathrm{flat}}
 (x;x_1,\ldots,x_N).
\]
Differentiating at $t=0$ yields the global
$\mathbf L_{-1}$ Ward identity
\begin{equation}\label{eq:global-L-minus-one}
 \left(
 \partial_x+\sum_{i=1}^N\partial_{x_i}
 \right)
 \mathcal G_\alpha^{\mathrm{flat}}=0.
\end{equation}
Consequently, \eqref{eq:level-two-bpz-amplitude} becomes the standard
BPZ equation
\begin{equation}\label{eq:level-two-bpz-standard}
 \left[
 \frac{1}{\alpha^2}\partial_x^2
 +\sum_{i=1}^N
 \left(
 \frac{1}{x-x_i}\partial_{x_i}
 +\frac{\Delta_{\beta_i}}{(x-x_i)^2}
 \right)
 \right]
 \mathcal G_\alpha^{\mathrm{flat}}(x;\mathbf x)=0.
\end{equation}
The conjugate null-vector relation gives the corresponding
anti-holomorphic equation.

\begin{corollary}[General degenerate insertions]
\label{cor:general-bpz}
Let $\alpha=\alpha_{r,s}$ and write the singular vector as in
\eqref{eq:null-vector-relation}. For a partition
$\nu=(\nu_1,\ldots,\nu_{\ell(\nu)})$, in the order fixed by
\eqref{eq:partition-operator}, set
\[
 \mathbf D_\nu^{(x)}
 :=\mathbf D_{\nu_1}^{(x)}\circ\cdots\circ
 \mathbf D_{\nu_{\ell(\nu)}}^{(x)},
\]
and define
\begin{equation}
 \mathcal L_{r,s}^{(x)}
 :=\sum_{|\nu|=rs}\sigma_\nu^{r,s}\mathbf D_\nu^{(x)}.
 \label{eq:general-bpz-operator}
\end{equation}
Then
\[
 \mathcal L_{r,s}^{(x)}\mathcal F_{\alpha_{r,s}}=
 \mathcal L_{r,s}^{(x)}\mathcal G_{\alpha_{r,s}}=0.
\]
The operator has order at most $rs$ with its principal part to be $\partial_x^{rs}$.
\end{corollary}

\begin{proof}
Repeated application of the round-disk Ward identity gives
\[
 \mathcal A(\mathbf L_{-\nu}\Psi_{\alpha_{r,s}})
 =\mathbf D_\nu^{(x)}\mathcal F_{\alpha_{r,s}}.
\]
Applying the amplitude to
$\mathbf S_{\alpha_{r,s}}\Psi_{\alpha_{r,s}}=0$ proves the first
identity, and the filling argument used above proves the second.
\end{proof}

\bibliography{cibib}
\bibliographystyle{alpha}

\end{document}